\documentclass[aps,prx,reprint,amsmath,amssymb,superscriptaddress,showkeys]{revtex4-1}

\usepackage[latin1]{inputenc}
\usepackage{graphicx}
\usepackage{amsmath}
\usepackage{amsfonts}
\usepackage{amsmath}
\usepackage{amssymb}
\usepackage{soul}
\usepackage{url}
\usepackage{bm}
\usepackage{amsthm}
\usepackage{mathtools}
\usepackage{dsfont}
\usepackage{stmaryrd}
\usepackage{adjustbox}
\usepackage{makecell}
\usepackage{xspace}
\usepackage{siunitx}
\usepackage{enumitem}
\usepackage{epigraph}
\usepackage{cancel}
\usepackage{lipsum}

\usepackage{filecontents}
\usepackage[dvipsnames]{xcolor}
\usepackage{hyperref}
\usepackage{cleveref}

\newcommand\myshade{85}
\colorlet{mylinkcolor}{BrickRed}
\colorlet{mycitecolor}{NavyBlue}
\colorlet{myurlcolor}{Aquamarine}

\hypersetup{
  linkcolor  = mylinkcolor!\myshade!black,
  citecolor  = mycitecolor!\myshade!black,
  urlcolor   = myurlcolor!\myshade!black,
  colorlinks = true,
}

\let\oldparagraph\paragraph
\renewcommand\paragraph[1]{\vskip0.5cm\oldparagraph{#1}}

\let\emptyset\varnothing

\newcommand\bX{\ensuremath{\bm X_t}\xspace}
\newcommand\bY{\ensuremath{\bm X_{t'}}}\xspace
\newcommand\Rk{\ensuremath{\mathcal{R}^{(k)}}}\xspace
\newcommand\Sk{\ensuremath{\mathcal{S}^{(k)}}}\xspace
\newcommand\Uk{\ensuremath{\mathcal{U}^{(k)}}}\xspace

\newcommand{\splitatcommas}[1]{%
\begingroup
\begingroup\lccode`~=`, \lowercase{\endgroup
    \edef~{\mathchar\the\mathcode`, \penalty0 \noexpand\hspace{0pt plus 1em}}%
    }\mathcode`,="8000 #1%
    \endgroup
}

\definecolor{darkgreen}{RGB}{50,200,8}

\graphicspath{ {./figs/} }

\newcommand\independent{\protect\mathpalette{\protect\independenT}{\perp}}
\def\independenT#1#2{\mathrel{\rlap{$#1#2$}\mkern2mu{#1#2}}}

\newtheorem{definition}{Definition}
\newtheorem{theorem}{Theorem}
\newtheorem{lemma}{Lemma}
\newtheorem{proposition}{Proposition}
\newtheorem{corollary}{Corollary}

\definecolor{BrewerRed}{RGB}{228,26,28}
\definecolor{BrewerGreen}{RGB}{77,175,74}
\definecolor{BrewerBlue}{RGB}{55,126,184}
\definecolor{BrewerPurple}{RGB}{152,78,163}
\definecolor{BrewerOrange}{RGB}{255,127,0}
\definecolor{BrewerYellow}{RGB}{215,215,41}
\definecolor{BrewerGrey}{RGB}{153,153,153}

\usepackage{filemod}
\newcommand{\includetikz}[2]{%
  \includegraphics{#2}
}

\usepackage{ctable}

\definecolor{lightblue}{RGB}{166,206,227}
\definecolor{darkblue}{RGB}{31 ,120,180}
\definecolor{lightred}{RGB}{251,154,153}
\definecolor{darkred}{RGB}{227,26,28}

\begin{document}
\title{Reconciling emergences: An information-theoretic approach \\
to identify causal emergence in multivariate data}

\author{Fernando E. Rosas} 
\thanks{F.R. and P.M. contributed equally to this work.\\E-mail: pam83@cam.ac.uk, f.rosas@imperial.ac.uk}
\affiliation{Center for Psychedelic Research, Department of Brain Science, Imperial College London, SW7 2DD, UK}
\affiliation{Data Science Institute, Imperial College London, London SW7 2AZ, UK}
\affiliation{Center for Complexity Science, Imperial College London, London SW7 2AZ, UK}

\author{Pedro A.M. Mediano}
\thanks{F.R. and P.M. contributed equally to this work.\\E-mail: pam83@cam.ac.uk, f.rosas@imperial.ac.uk}
\affiliation{Department of Psychology, University of Cambridge, Cambridge CB2 3EB, UK}

\author{Henrik J. Jensen}
\affiliation{Center for Complexity Science, Imperial College London, London SW7 2AZ, UK}
\affiliation{Department of Mathematics, Imperial College London, London SW7 2AZ, UK}
\affiliation{Institute of Innovative Research, Tokyo Institute of Technology, Yokohama 226-8502, Japan}

\author{Anil~K.~Seth}
\affiliation{Sackler Center for Consciousness Science, Department of Informatics, University of Sussex, Brighton BN1 9RH, UK}
\affiliation{CIFAR Azrieli Program on Brain, Mind, and Consciousness, Toronto M5G 1M1, Canada}

\author{Adam B. Barrett}
\affiliation{Sackler Center for Consciousness Science, Department of Informatics, University of Sussex, Brighton BN1 9RH, UK}
\affiliation{The Data Intensive Science Centre, Department of Physics and Astronomy, University of Sussex, Brighton BN1 9QH, UK}

\author{Robin L. Carhart-Harris} 
\affiliation{Center for Psychedelic Research, Department of Brain Science, Imperial College London, London SW7 2DD, UK}

\author{Daniel Bor}
\affiliation{Department of Psychology, University of Cambridge, Cambridge CB2 3EB, UK}

\begin{abstract}

\noindent The broad concept of \emph{emergence} is instrumental in various of
the most challenging open scientific questions -- yet, few quantitative
theories of what constitutes emergent phenomena have been proposed. This
article introduces a formal theory of causal emergence in multivariate systems,
which studies the relationship between the dynamics of parts of a system and
macroscopic features of interest. Our theory provides a quantitative definition
of \emph{downward causation}, and introduces a complementary modality of
emergent behaviour -- which we refer to as \emph{causal decoupling}. Moreover,
the theory allows practical criteria that can be efficiently calculated in
large systems, making our framework applicable in a range of scenarios of
practical interest. We illustrate our findings in a number of case studies,
including Conway's Game of Life, Reynolds' flocking model, and neural activity
as measured by electrocorticography.

\end{abstract}

\maketitle

While most of our representations of the physical world are hierarchical, there
is still no agreement on how the co-existing ``layers'' of this hierarchy
interact. On the one hand, \textit{reductionism} claims that all levels can
always be explained based on sufficient knowledge of the lowest scale and,
consequently -- taking an intentionally extreme example -- that a sufficiently
accurate theory of elementary particles should be able to predict the existence
of social phenomena like communism. On the other hand, \textit{emergentism}
argues that there can be autonomy between layers, i.e. that some phenomena at
macroscopic layers might only be accountable in terms of other macroscopic
phenomena. While emergentism might seem to better serve our intuition, it is
not entirely clear how a rigorous theory of emergence could be formulated
within our modern scientific worldview, which tends to be dominated by
reductionist principles.

Emergent phenomena are usually characterised as either strong or
weak~\cite{routledge2019}. \emph{Strong emergence} corresponds to the somewhat
paradoxical case of supervenient properties with irreducible causal
power~\cite{bedau2002downward}; i.e. properties that are fully determined by
microscopic levels but can nevertheless exert causal influences that are not
entirely accountable from microscopic considerations~\footnote{The case of
strong emergence most commonly argued in the literature is the one of conscious
experiences with respect to their corresponding physical
substrate~\cite{chalmers2006strong,chang2019information}.}. Strong emergence
has been as much a cause of wonder as a perennial source of philosophical
headaches, being described as ``uncomfortably like magic'' while accused of
being logically inconsistent ~\cite{bedau2002downward} and sustained on
illegitimate metaphysics~\cite{bedau1997weak}. \emph{Weak emergence} has been
proposed as a more docile alternative to strong emergence, where macroscopic
features have irreducible causal power in practice but not in principle. A
popular formulation of weak emergence is due to Bedau~\cite{bedau1997weak}, and
corresponds to properties generated by elements at microscopic levels in such
complicated ways that they cannot be derived via explanatory shortcuts, but
only by exhaustive simulation. While this formulation is usually accepted by
the scientific community, it is not well-suited to address mereological
questions about emergence in scenarios where parts-whole relationships are the
primary interest.

Part of the difficulty in building a deeper understanding of strong emergence
is the absence of simple but clear analytical models that can serve the
community to guide discussions and mature theories. Efforts have been made to
introduce quantitative metrics of weak emergence~\cite{seth2010measuring},
which enable fine-grained data-driven alternatives to traditional all-or-none
classifications. In this vein, an attractive alternative comes from the work on
\emph{causal emergence} introduced in Ref.~\cite{hoel2013quantifying} and later
developed in Refs.~\cite{hoel2017map,klein2020emergence}, which showed that
macroscopic observables can sometimes exhibit more causal power (as understood
within the framework of Pearl's \textit{do-calculus}~\cite{pearl2000causality})
than microscopic variables. However, this framework relies on strong
assumptions that are rarely satisfied in practice, which severely hinders its
applicability~\footnote{This point is further ellaborated in
Section~\ref{sec:hoel}}.

Inspired by Refs.~\cite{seth2010measuring,hoel2013quantifying}, here we
introduce a practically useful and philosophically innocent framework to study
causal emergence in multivariate data. Building on previous
work~\cite{rosas2018selforg}, we take the perspective of an experimentalist who
has no prior knowledge of the underlying phenomenon of interest, but has
sufficient data of all relevant variables that allows an accurate statistical
description of the phenomenon. In this context, we put forward a formal
definition of causal emergence that doesn't rely on coarse-graining functions
as Ref.~\cite{hoel2013quantifying}, but addresses the ``paradoxical''
properties of strong emergence based on the laws of information flow in
multivariate systems.

The main contribution of this work is to enable a rigorous, quantitative
definition of \emph{downward causation}, and introduce a novel notion of
\emph{causal decoupling} as a complementary modality of causal emergence.
Another contribution is to extend the domain of applicability of causal
emergence analyses to include cases of observational data, in which case
causality ought to be understood in the Granger sense, i.e. as predictive
ability~\cite{bressler2011wiener}. Furthermore, our framework yields practical
criteria that can be effectively applied to large systems, bypassing
prohibitive estimation issues that severely restrict previous approaches.

The rest of this paper is structured as follows. First,
Section~\ref{sec:examples} introduces some fundamental intuitions by discussing
minimal examples of emergence. Then, Section~\ref{sec:formal_theory} presents
the core of our theory, and Section~\ref{sec:measuring} discusses practical
methods to measure emergence from experimental data. Our framework is then
illustrated on a number of case studies, presented in
Section~\ref{sec:case_studies}. Finally, Section~\ref{sec:discussion} concludes
the paper with a discussion of some of the implications of our findings.

\section{Fundamental intuitions}
\label{sec:examples}

To ground our intuitions, let us introduce minimal examples that embody a few
key notions of causally emergent behaviour. Throughout this section, we
consider systems composed of $n$ parts described by a binary vector $\bm X_t =
(X_t^1,\dots,X_t^n) \in\{0,1\}^n$, which undergo Markovian stochastic dynamics
following a transition probability $p_{\bm X_{t+1}|\bm X_t}$. For simplicity,
we assume that at time $t$ the system is found in an entirely random
configuration (i.e. $p_{\bm X_t}(\bm x_t) =2^{-n}$). From there, we consider
three evolution rules.

\textbf{Example 1:} Consider a temporal evolution where the parity of $\bm X_t$
is preserved with probability $\gamma \in (0, 1)$. Mathematically,
\begin{equation}
p_{\bm X_{t+1}|\bm X_t}(\bm x_{t+1}|\bm x_t) = 
\begin{cases}
\frac{\gamma}{2^{n-1}} & \small{\text{if } \oplus_{j=1}^n x_{t+1}^j =  \oplus_{j=1}^n x_{t}^j ,} \\ 
\frac{1-\gamma}{2^{n-1}} & \small{\text{otherwise,}} ~
\end{cases}\nonumber
\end{equation}
for all $t\in\mathbb{N}$, where $\oplus_{j=1}^n a_j \coloneqq 1$ if
$\sum_{j=1}^n a_j$ is even and zero otherwise. Put simply: $\bm x_{t+1}$ is a
random sample from the set of all strings with the same parity as $\bm x_t$
with probability $\gamma$; and is a sample from the strings with opposite
parity with probability $1-\gamma$.

This evolution rule has a number of interesting properties. First, the system
has a non-trivial causal structure, since some properties of the future state
(its parity) can be predicted from the past state. However, this structure is
noticeable \emph{only} at the collective level, as no individual variable has
any predictive power over the evolution of itself or any other variable (see
Figure~\ref{fig:examples}). Furthermore, even the complete past of the system
$\bm X_t$ has no predictive power over any individual future $X_{t+1}^j$. This
case shows an extreme kind of causal emergence that we call ``causal
decoupling,'' in which the parity predicts its own evolution but no element (or
subset of elements) predicts the evolution of any other element.

\begin{figure}[t]
  \centering
  \includetikz{tikz/}{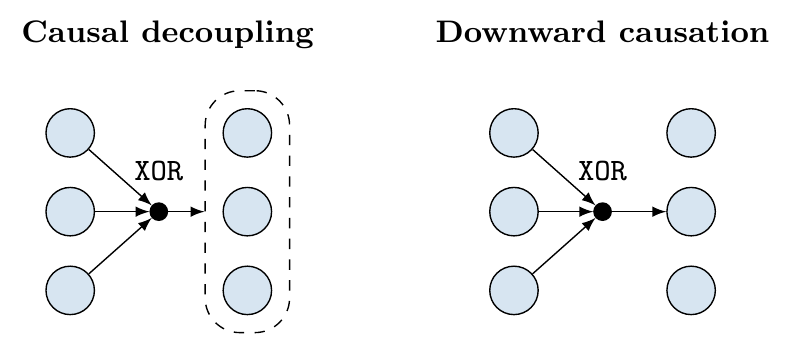}
  \caption{\textbf{Minimal examples of causally emergent dynamics}. In Example~1
  (\emph{left}) the system's parity tends to be preserved while no interactions
  occur between low-level elements, which is an example of causal decoupling. In
  Example~2 (\emph{right}) the system's parity determines one element only,
  corresponding to downward causation.}
  \label{fig:examples}
\end{figure}

\textbf{Example 2:} Consider now a system where the parity of $\bX$ determines
$X_{t+1}^1$ (i.e. $X^1_{t+1} = \oplus_{i=1}^n X^i_t$), and $X^j_{t+1}$ for
$j\neq 1$ is a fair coin flip independent of $\bX$ (see
Figure~\ref{fig:examples}). In this scenario $\bX$ predicts $X_{t+1}^1$ with
perfect accuracy, while it can be verified that $X_t^i \independent X_{t+1}^1$
for all $i\in\{1,\dots,n\}$. Therefore, under this evolution rule the whole
system has a causal effect over a particular element, although this effect
cannot be attributed to any individual part~\footnote{A related discussion
about ``synergistic information flow'' can be found in
Ref.~\cite{james2018modes}.}, being a minimal example of downward causation.

\textbf{Example 3:} Let us now study an evolution rule that includes the
mechanisms of both Examples~1 and 2. Concretely, consider
\begin{equation}
p_{\bm X_{t+1}|\bm X_t}(\bm x_{t+1}|\bm x_t) = 
\begin{cases}
0 & \small{ \text{if } x_{t+1}^1 \neq \oplus_{j=1}^n x_{t}^j, }\\
\frac{\gamma}{2^{n-2}} & \small{ \text{if } x_{t+1}^1 = \oplus_{j=1}^n x_{t}^j }\\
&\small{ \text{and } \oplus_{j=1}^n x_{t+1}^j =  \oplus_{j=1}^n x_{t}^j, }\\ 
\frac{1-\gamma}{2^{n-2}} & \small{ \text{otherwise.} }~
\end{cases}\nonumber
\end{equation}
As in Example~1, the parity of $\bX$ is transfered to $\bm X_{t+1}$ with
probability $\gamma$; additionally, it is guaranteed that $X_{t+1}^1 =
\oplus_{i=1}^n X_t^i$. Hence, in this case not only is there a macroscopic
effect that cannot be explained from the parts, but at the same time there is
another effect going from the whole to one of the parts. Importantly, both
effects \emph{co-exist} independently of each other.

The above are minimal examples of dynamical laws that cannot be traced from the
interactions between their elementary components: Example~1 shows how a
collective property can propagate without interacting with its underlying
substrate; Example~2 how a collective property can influence the evolution of
specific parts; and Example~3 how these two kinds of phenomena take place in
the same system. All these issues are formalised by the theory developed in the
next section.

\section{A formal theory of causal emergence}
\label{sec:formal_theory}

This section presents the main body of our theory of causal emergence. To fix
ideas, we consider a scientist measuring a system composed of $n$ parts. The
scientist is assumed to measure the system regularly over time, and the results
of those measurements are denoted by $\bm X_t = (X_t^1,\dots,X_t^n)$, with
$X_t^i\in \mathcal{X}_i$ corresponding to the state of the
$i$\textsuperscript{th} part at time $t\in \mathbb{N}$ with phase space
$\mathcal{X}_i$. When referring to a collection of parts, we use the notation
$\bm X_t^{\alpha} = (X_t^{i_1},\dots,X_t^{i_K})$ for $\alpha =
\{i_1,\dots,i_K\}\subset \{1,\dots,n\}$. We also use the shorthand notation
$[n] \coloneqq \{1,\dots,n\}$.

Our analysis considers two time points of the evolution of the system, denoted
as $t$ and $t'$, with $t< t'$. The corresponding dynamics are encoded in the
transition probability $p_{\bm X_{t'}|\bm X_t}(\bm x_{t'}|\bm x_t)$. We
consider features $V_t\in \mathcal{V}$ generated via a conditional probability
$p_{V_t|\bm X_t}$ that are \emph{supervenient} on the underlying system; i.e.
that does not provide any predictive power for future states at times $t'>t$ if
the complete state of the system at time $t$ is known with perfect precision.
We formalise this condition by requiring $V_t$ to be statistically independent
of $\bm X_{t'}$ when $\bm X_t$ is given (equivalently, that $V_t - \bm X_t -
\bm X_{t'}$ form a Markov chain), for all $t'> t$. This includes as
particular cases deterministic functions $F:\prod_{j=1}^n \mathcal{X}_j\to
\mathcal{V}$ such that $V_t = F(\bm X_t)$, as well as aggregate properties
affected by observational noise.

\begin{figure}[ht]
   \centering
   \includetikz{tikz/}{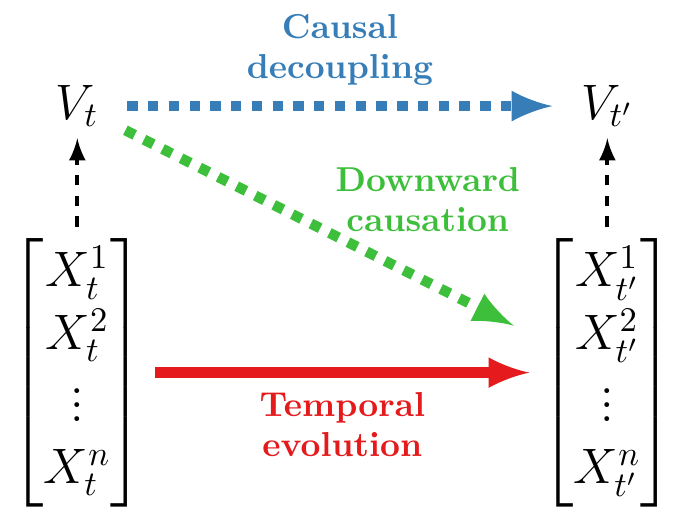}
   \caption{\textbf{Diagram of causally emergent relationships.} Causally emergent
   features have predictive power beyond individual components. Downward causation 
   takes place when that predictive power refers to individual elements; causal decoupling
   when it refers to itself or other high-order features.}
   \label{fig:lattice}
\end{figure}

\subsection{Partial information decomposition}
\label{sec:pid}

Our theory is based on the \emph{Partial Information Decomposition} (PID)
framework~\cite{williams2010nonnegative}, which provides powerful tools to
reason about information in multivariate systems. In a nutshell, PID decomposes
the information that $n$ sources $\bm X = (X^1,\dots,X^n)$ provide about a
target variable $Y$ in terms of information atoms as follows:
\begin{equation}
I(\bm X ; Y) = \sum_{\bm\alpha\in\mathcal{A}} I_\partial^{\bm\alpha}(\bm X ; Y)~,
\end{equation}
with $\mathcal{A} = \{ \{\alpha_1,\dots,\alpha_L\} :
\alpha_i\subseteq[n], \alpha_i \not\subset \alpha_j \forall i,j \}$ being the
set of antichain collections~\cite{williams2010nonnegative}. Intuitively,
$I_\partial^{\bm\alpha}$ for $\bm\alpha=\{\alpha_1,\dots,\alpha_L\}$ 
represents the information that the collection of
variables $\bm X^{\alpha_1},\dots,\bm X^{\alpha_L}$ provide redundantly, but
their sub-collections don't. For example, for $n=2$ source variables, $\bm\alpha
= \{\{1\}\{2\}\}$ corresponds to the information about $Y$ that is provided
independently (and, hence, redundantly) by both of them, $\bm\alpha =
\{\{i\}\}$ to the information provided uniquely by $X_i$, and, most
interestingly, $\bm\alpha = \{\{12\}\}$ corresponds to the information
provided by both sources jointly but not separately -- commonly referred to as
\emph{informational synergy}.

One of the drawbacks of PID is that the number of atoms (i.e. the cardinality
of $\mathcal{A}$) grows super-exponentially with the number of sources, and
hence it is useful to coarse-grain the decomposition according to specific
criteria. Here we introduce the notion of \emph{$k$\textsuperscript{th}-order
synergy} between $n$ variables, which is calculated as
\begin{equation}
  \textnormal{\texttt{Syn}}^{(k)}(\bm X ; Y)\coloneqq \sum_{\bm\alpha\in\mathcal{S}^{(k)}} I_\partial^{\bm\alpha}(\bm X ; Y)~,\nonumber
\end{equation}
with $\mathcal{S}^{(k)} = \big\{\{\alpha_1,\dots,\alpha_L\} \in\mathcal{A} :
\min_j |\alpha_j| > k \big\}$. Intuitively, $\texttt{Syn}^{(k)}(\bm X; Y)$
corresponds to the information about the target that is provided by the whole
$\bm X$ but is not contained in any set of $k$ or less parts when considered
separately from the rest. Accordingly, $\mathcal{S}^{(k)}$ only contains
collections with groups of more than $k$ sources.

Similarly, we introduce the unique information of $\bm X^\beta$ with
$\beta\subset [n]$ with respect to sets of at most $k$ other variables, which
is calculated as
\begin{equation}
  \textnormal{\texttt{Un}}^{(k)}(\bm X^{\beta} ; Y | \bm X^{-\beta} ) \coloneqq \sum_{\bm \alpha \in \mathcal{U}^{(k)}(\beta)} I_\partial^{ \bm\alpha}(\bm X ; Y)~. \nonumber
\end{equation}
Above, $\mathcal{U}^{(k)}(\beta) = \{ \bm \alpha\in \mathcal{A} : \beta \in
\bm\alpha, \forall \alpha\neq\beta \in\bm\alpha, \alpha\subseteq
[n]\setminus\beta, |\alpha| > k\}$, and $\bm X^{-\beta}$ being all the
variables in $\bm X$ whose indices are not in $\beta$. Put simply,
$\texttt{Un}^{(k)}(\bX^\beta; Y | \bm X^{-\beta})$ represents the information
carried by $\bm X^\beta$ about $Y$ that no group of $k$ or less variables
within $\bm X^{-\bm\beta}$ has on its own. Note that these coarse-grained terms
can be used to build a general decomposition of $I(\bm X, Y)$ described in
Appendix~\ref{app:pid}, the properties of which are proven in
Appendix~\ref{app:highorderPID}.

One peculiarity of PID is that it postulates the structure of information atoms
and the relations between them, but it does not prescribe a particular
functional form to compute $I_\partial^{\bm\alpha}$~\footnote{Only one of the
information atoms must be specified to determine the whole PID -- usually the
redundancy between all individual elements~\cite{williams2010nonnegative}.}.
Please note that there have been multiple proposals for specific functional
forms of $I_\partial^{\bm\alpha}$ in the PID literature; see e.g.
Refs.~\cite{ay2019information,lizier2018information,james2018unique,ince2017measuring}.
A particular method for fully computing the information atoms based on a recent
PID~\cite{rosas2020operational} is discussed in
Sec.~\ref{sec:quantifying_synergies}.

Conveniently, our theory doesn't rely on a specific functional form of PID, but
only on a few basic properties that are precisely formulated in
Appendix~\ref{app:highorderPID}. Therefore, the theory can be instantiated
using any PID -- as long as those properties are satisfied. Importantly, as
shown in Section~\ref{sec:general_tools}, the theory allows the derivation of
practical metrics that are valid independently of the PID chosen.

\subsection{Defining causal emergence}\label{sec:defining}

With the tools of PID at hand, now we introduce our formal definition of causal
emergence.

\begin{definition}\label{def:emergence_unique}
For a system described by $\bm X_t$, a supervenient feature $V_t$
is said to exhibit causal emergence of order $k$ if
\begin{equation}\label{eq:irreducible_causal}
\textnormal{\texttt{Un}}^{(k)}(V_t;\bm X_{t'} | \bm X_t) > 0~.
\end{equation}
\end{definition}

Accordingly, causal emergence takes place when a supervenient feature $V_t$ has
irreducible causal power, i.e. when it \textit{exerts causal influence that is
not mediated by any of the parts of the system}. In other words, $V_t$
represents some emergent collective property of the system if: 1) contains
information that is dynamically relevant (in the sense that it predicts the
future evolution of the system); and 2) this information is beyond what is
given by the groups of $k$ parts in the system when considered separately.

To better understand the implications of this definition, let us study some
of its basic properties.
\begin{lemma}\label{lemma:basic_props}
Consider a feature $V_t$ that exhibits causal emergence of order 1 over $\bm X_t$. Then,
\begin{itemize}
\item[(i)] The dimensionality of the system satisfies $n\geq 2$.
\item[(ii)] There exists no deterministic function $g(\cdot)$ such that $V_{t}=g(X_t^j)$ for any $j=1,\dots,n$.
\end{itemize} 
\end{lemma}
\begin{proof}
See Appendix~\ref{app:proofs}.
\end{proof}

These two properties establish causal emergence as a \emph{fundamentally
collective} phenomenon. In effect, property (i) states that causal emergence is
a property of multivariate systems, and property (ii) that $V_t$ cannot have
emergent behaviour if it can be perfectly predicted from a single variable.

In order to use Definition~\ref{def:emergence_unique}, one needs a candidate
feature $V_t$ to be tested. However, in some cases there are no obvious
candidates for an emergent feature, for which
Definition~\ref{def:emergence_unique} might seem problematic. Our next result
provides a criterion for the existence of emergent features based solely on the
system's dynamics.
\begin{theorem}\label{prop:emergence_syn}
A system $\bX$ has a causally emergent feature of order $k$ if and only if
\begin{equation}
  \textnormal{\texttt{Syn}}^{(k)}(\bm X_t;\bm X_{t'})>0 ~ .
\end{equation}
\end{theorem}
\begin{proof}
See Appendix~\ref{app:highorderPID}.
\end{proof}
\begin{corollary}\label{cor:syn_cap}
The following bound holds for any supervenient feature $V_t$:
$\textnormal{\texttt{Un}}^{(k)}(V_t;\bm X_{t'} | \bm X_t) \leq
\textnormal{\texttt{Syn}}^{(k)}(\bm X_t ; \bm X_{t'})$.
\end{corollary}

This result shows that the capability of exhibiting emergence is closely
related to how synergistic the system components are with respect to their
future evolution. Importantly, this result enables us to determine whether or
not the system admits any emergent features by just inspecting the synergy
between its parts -- \emph{without knowing what those features might be}.
Conversely, this result also allows us to discard the existence of causal
emergence by checking a single condition: the lack of dynamical synergy.
Furthermore, Corollary~\ref{cor:syn_cap} implies that the quantity
$\textnormal{\texttt{Syn}}^{(k)}(\bm X_t;\bm X_{t'})$ serves as a measure of
the \emph{emergence capacity} of the system, as it upper-bounds the unique
information of all possible supervenient features.

Theorem~\ref{prop:emergence_syn} establishes a direct link between causal
emergence and the system's statistics, avoiding the need for the observer to
propose a particular feature of interest. It is important to remark that the
emergence capacity of a system depends on the system's partition into
microscopic elements~\footnote{In effect, it is plausible that a system might
have emergence capacity under one microscopic representation, but not with
respect to another after a change of variables.}. Therefore, emergence in the
context of our theory always refers to ``emergence with respect to a given
microscopic partition.''

\subsection{A taxonomy of emergence}
\label{sec:taxonomy}

Our theory, so far, is able to detect \emph{whether} there is emergence taking
place; the next step is to be able to characterise \textit{which kind} of
emergence it is. For this purpose, we combine our feature-agnostic criterion of
emergence presented in Theorem~\ref{prop:emergence_syn} with Integrated
Information Decomposition, $\Phi$ID, a recent extension of PID to multi-target
settings~\cite{mediano2019beyond}.

Using $\Phi$ID, one can decompose a PID atom as
\begin{equation}\label{eq:PhiID_dec}
I_\partial^{\bm\alpha}(\bX ; \bm X_{t'}) = \sum_{\bm\beta\in\mathcal{A}} I_\partial^{\bm\alpha\rightarrow \bm\beta}(\bX ; \bm X_{t'}) ~.
\end{equation}
For example, if $n=2$ then $I_\partial^{\{1\}\{2\}\rightarrow\{1\}\{2\}}$
represents the information shared by both time series at both timesteps (for
example, when $X^1_t,X^2_t,X^1_{t'},X^2_{t'}$ are all copies of each other);
and $I_\partial^{\{12\}\rightarrow\{1\}}$ corresponds to the synergistic causes
in \bX that have a unique effect on $X^1_{t'}$ (for example, when $X^1_{t'} =
X^1_t \oplus X^2_t$). More details and intuitions on $\Phi$ID can be found in
Ref.~\cite{mediano2019beyond}.

With the fine-grained decomposition provided by $\Phi$ID one can discriminate
between different kinds of synergies. In particular, we introduce the
\emph{downward causation} and \emph{causal decoupling indices of order $k$},
denoted by $\mathcal{D}^{(k)}$ and $\mathcal{G}^{(k)}$ respectively, as
\begin{align}
\mathcal{G}^{(k)}(\bX;\bm X_{t'}) \coloneqq & \sum_{\substack{\bm\alpha,\bm\beta\in\mathcal{S}^{(k)}}} I_\partial^{\bm\alpha\rightarrow\bm\beta} (\bX;\bm X_{t'}) ~, \label{eq:decoupling_index}\\
\mathcal{D}^{(k)} (\bX;\bm X_{t'}) \coloneqq & \sum_{\substack{\bm\alpha\in\mathcal{S}^{(k)}\\\bm\beta\in\mathcal{A}\setminus\mathcal{S}^{(k)}}} I_\partial^{\bm\alpha\rightarrow\bm\beta} (\bX;\bm X_{t'})~. \label{eq:downward_index}
\end{align}
From these definitions and Eq.~\eqref{eq:PhiID_dec}, one can verify that
\begin{equation}\label{eq:syn_decomp_sec}
\texttt{Syn}^{(k)} (\bX;\bm X_{t'}) = \mathcal{G}^{(k)} (\bX;\bm X_{t'}) + \mathcal{D}^{(k)} (\bX;\bm X_{t'})~. 
\end{equation}

\begin{figure}[t]
  \centering
  \includetikz{tikz/}{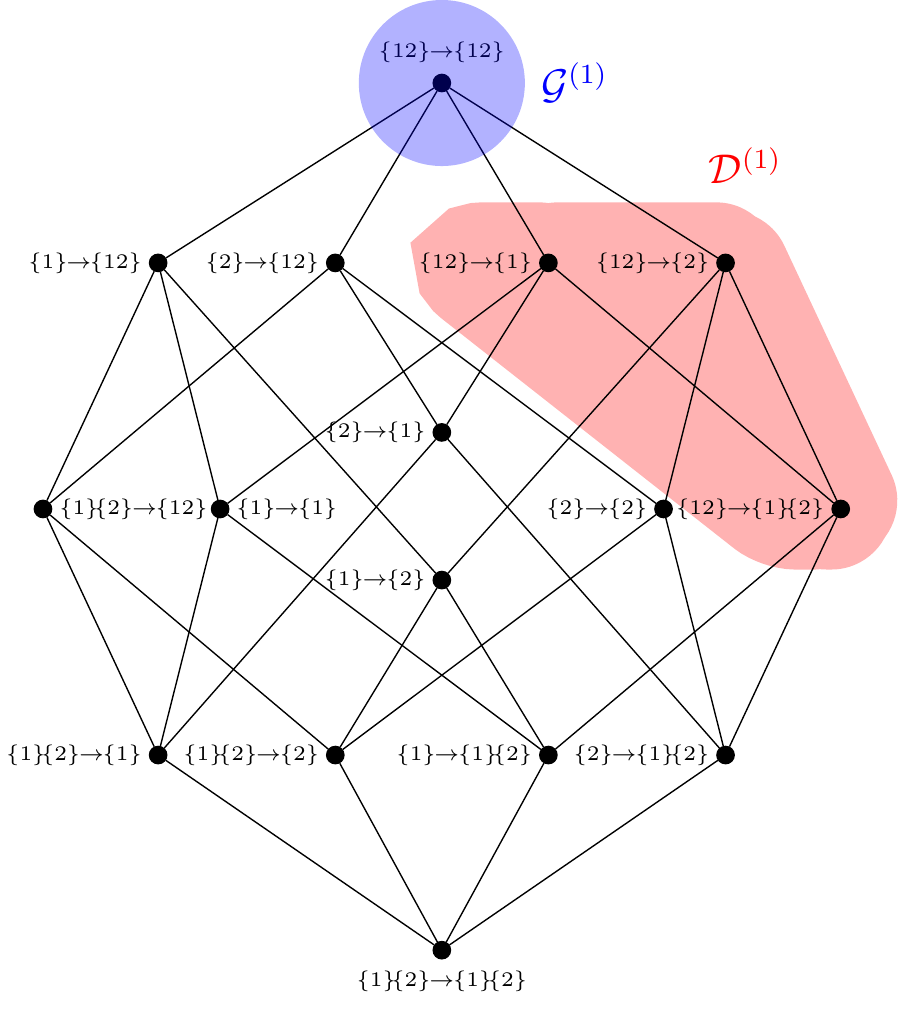}

  \caption{\textbf{Integrated Information Decomposition} ($\Phi$ID). $\Phi$ID
  lattice for $n=2$ time series~\cite{mediano2019beyond}, with downward ($\mathcal{D}$)
  causation and causal decoupling ($\mathcal{G}$) terms highlighted.}
  \label{fig:causallattice}
\end{figure}

Therefore, the emergence capacity of a system naturally decomposes in two
different components: information about $k$-plets of future variables, and
information about future collective properties beyond $k$-plets. The $\Phi$ID
atoms that belong to these two terms are illustrated within the $\Phi$ID
lattice for two time series in Figure~\ref{fig:causallattice}. The rest of this
section shows that $\mathcal{D}^{(k)}$ and $\mathcal{G}^{(k)}$ are natural
metrics of downward causation and causal decoupling, respectively.

\subsubsection{Downward causation}

Intuitively, downward causation occurs when collective properties have
irreducible causal power over individual parts. More formally:

\begin{definition}
A supervenient feature $V_t$ exhibits downward causation of order $k$ if, for
some $\alpha$ with $|\alpha| = k$:
\begin{equation}
  \textnormal{\texttt{Un}}^{(k)}(V_t; \bm X_{t'}^{\alpha} | \bX ) > 0~.
\end{equation}
 \end{definition}

Note that, in contrast with Definition~\ref{def:emergence_unique}, downward
causation requires the feature $V_t$ to have unique predictive power over the
evolution of specific subsets of the whole system. In particular, an emergent
feature $V_t$ that has predictive power over e.g. $X_{t'}^j$ is said to exert
downward causation, as it predicts something about $X_{t'}^j$ that could not be
predicted from any particular $X_t^i$ for $i\in[n]$. Put differently, in a
system with downward causation the whole has an effect on the parts that cannot
be reduced to low-level interactions. A minimal case of this is provided by
Example~2 in Section~\ref{sec:examples}.

Our next result formally relates downward causation with the index
$\mathcal{D}^{(k)}$ introduced in Eq.~\eqref{eq:downward_index}.

\begin{theorem}\label{prop:downward}
A system $\bX$ admits features that exert downward causation of order $k$ iff
$\mathcal{D}^{(k)}(\bm X_t; \bm X_{t'}) >0$.
\end{theorem}
\begin{proof}
See Appendix~\ref{app:proofs}.
\end{proof}

\subsubsection{Causal decoupling} \label{sec:decoupling}

In addition to downward causation, causal decoupling takes place when
collective properties have irreducible causal power over other collective
properties. In technical terms:
\begin{definition}
A supervenient feature $V_t$ is said to exhibit \textit{causal decoupling} of
order $k$ if
\begin{equation}
\textnormal{\texttt{Un}}^{(k)}(V_t; V_{t'} | \bX, \bY) > 0~.
\end{equation}
Furthermore, $V_t$ is said to have \textit{pure causal decoupling} if
$\textnormal{\texttt{Un}}^{(k)}(V_t;\bY | \bX) > 0$ and
$\textnormal{\texttt{Un}}^{(k)}(V_t;\bm X^{\alpha}_{t'} | \bX) = 0$ for all
$\alpha\subset [n]$ with $|\alpha|=k$. Finally, a system is said to be
\textit{perfectly decoupled} if all the emergent features exhibit pure causal
decoupling.
\end{definition}

Above, the term $\textnormal{\texttt{Un}}^{(k)}(V_t; V_{t'} | \bX, \bY)$ refers
to information that $V_t$ and $V_{t'}$ share that cannot be found in any
microscopic element, either at time $t$ or $t'$. Note that features that
exhibit causal decoupling could still exert influence over the evolution of
individual elements, while features that exhibit pure decoupling cannot. In
effect, the condition $\textnormal{\texttt{Un}}(V_t;X^j_{t'} | \bX) = 0$
implies that the high-order causal effect does not affect any particular part
-- only the system as a whole. Interestingly, a feature that exhibits pure
causal decoupling can be thought of as having ``a life of its own;'' a sort of
\emph{statistical ghost}, that perpetuates itself over time without any
individual part of the system influencing or being influenced by it. The
system's parity, in the first example of Sec.~\ref{sec:examples}, constitutes a
simple example of perfect causal decoupling.

We close this section by formally establishing the connection between causal
decoupling and the index $\mathcal{G}^{(k)}$ introduced in
Eq.~\eqref{eq:decoupling_index}.

\begin{theorem}\label{prop:decoupled}
A system possesses features that exhibit causal decoupling if and only if
$\mathcal{G}^{(k)}(\bm X_t; \bm X_{t'}) > 0$. Additionally, the system is
perfectly decoupled if $\mathcal{G}^{(k)}(\bm X_t; \bm X_{t'}) > 0$ and
$\mathcal{D}^{(k)}(\bm X_t; \bm X_{t'}) = 0$.
\end{theorem}
\begin{proof}
See Appendix~\ref{app:proofs}.
\end{proof}

\section{Measuring emergence}
\label{sec:measuring}

This section explores methods to operationalise the framework presented in the
previous section. We discuss two approaches: first,
Section~\ref{sec:general_tools} introduces sufficiency criteria that are
practical for use in large systems; then,
Section~\ref{sec:quantifying_synergies} illustrates how further considerations
can be made if one adopts a specific method of computing $\Phi$ID atoms. The
latter approach provides accurate discrimination at the cost of being
data-intensive and hence only applicable to small systems; the former can be
computed in large systems and its results hold independently of the chosen PID,
but is vulnerable to misdetections (i.e. false negatives).

\subsection{Practical criteria for large systems}\label{sec:general_tools}

While theoretically appealing, our proposed framework suffers from the
challenge of estimating joint probability distributions over many random
variables, and the computation of the $\Phi$ID atoms themselves. As an
alternative, we consider approximation techniques that do not require the
adoption of any particular PID or $\Phi$ID function and are data-efficient,
since they are based on pairwise distributions only.

As practical criteria to measure causal emergence of order $k$, we introduce
the quantities $\Psi_{t,t'}^{(k)}$, $\Delta_{t,t'}^{(k)}$, and
$\Gamma_{t,t'}^{(k)}$. For simplicity, we write here the special case $k=1$,
and provide full formulae for arbitrary $k$ and accompanying proofs in
Appendix~\ref{app:proofs2}:
\begin{subequations}
\begin{align}
\Psi_{t,t'}^{(1)}(V)&\coloneqq I(V_t; V_{t'}) - \sum_{j} I(X_{t}^{j}; V_{t'} )~, \\
\Delta_{t,t'}^{(1)}(V)&\coloneqq \max_{j} \left( I(V_t; X^{j}_{t'}) - \sum_{i} I(X_{t}^{i}; X^{j}_{t'} ) \right)~, \\
\Gamma_{t,t'}^{(1)}(V)&\coloneqq \max_{j} I(V_t; X_{t'}^{j})~.
\end{align}
\label{eq:emergence_psi}%
\end{subequations}

Our next result links these quantities with the formal definitions in
Section~\ref{sec:formal_theory}, showing their value as practical criteria to
detect causal emergence.

\begin{proposition}\label{prop:emergence_psi}
$\Psi_{t,t'}^{(k)}(V) > 0$ is a sufficient condition for $V_t$ to be causally
emergent. Similarly, $\Delta_{t,t'}^{(k)}(V) > 0$ is a sufficient condition for
$V_t$ to exhibit downward causation. Finally, $\Psi_{t,t'}^{(k)}(V) > 0$ and
$\Gamma_{t,t'}^{(k)}(V) = 0$ is sufficient for causal decoupling.
\end{proposition}
\begin{proof}
See Appendix~\ref{app:proofs2}.
\end{proof}

Although calculating whether a system has emergent features via
Proposition~\ref{prop:emergence_syn} may be computationally challenging, if one
has a candidate feature $V$ one believes may be emergent, one can compute the
simple quantities in Eqs.~\eqref{eq:emergence_psi} which depend only on
standard mutual information and bivariate marginals, and scales linearly with
system size (for $k=1$). These quantities are easy to compute and test for
significance using standard information-theoretic
tools~\cite{kraskov2004estimating,lizier2014jidt}. Moreover, the outcome of
these measures is valid for any choice of PID and $\Phi$ID that is compatible
with the properties specified in Appendix~\ref{app:highorderPID}.

In a broader context, $\Psi_{t,t'}^{(k)}$ and $\Delta_{t,t'}^{(k)}$ belong to
the same \emph{whole-minus-sum} family of measures as the interaction
information~\cite{mcgill1954multivariate,williams2010nonnegative}, the
redundancy-synergy index~\cite{timme2014synergy} and, more recently, the
O-information $\Omega$~\cite{rosas2019quantifying} -- which cannot measure
synergy by itself, but only the balance between synergy and redundancy. In
practice, this means that if there is redundancy in the system it will be
harder to detect emergence, since redundancy will drive $\Psi_{t,t'}^{(k)}$ and
$\Delta_{t,t'}^{(k)}$ more negative. Furthermore, by summing all marginal
mutual informations (e.g. $I(X_t^j; V_{t'})$ in the case of
$\Psi_{t,t'}^{(1)}$), these measures effectively \emph{double-count} redundancy
up to $n$ times, further penalising the criteria. This problem of
double-counting can be avoided if one is willing to commit to a particular PID
or $\Phi$ID function, as we show next.

\subsection{Measuring emergence via synergistic channels}
\label{sec:quantifying_synergies}

This section leverages recent work on information decomposition reported in
Ref.~\cite{rosas2020operational}, and presents a way of directly measuring the
emergence capacity and the indices of downward causation and causal decoupling.
The key takeaway of this section is that if one adopts a particular $\Phi$ID,
then it is possible to evaluate $\mathcal{D}^{(k)}$ and $\mathcal{G}^{(k)}$
directly, providing a direct route to detect emergence without
double-counting redundancy, as the methods introduced in
Section~\ref{sec:general_tools} do. Moreover, additional properties may 
become available due to the characteristics of the particular $\Phi$ID chosen.

Let us first introduce the notion of $k$-synergistic channels: mappings $p_{V|
\bm X}$ that convey information about $\bm X$ but not about any of the parts
$\bm X^{\alpha}$ for all $|\alpha|=k$. The set of all $k$-synergistic channels
is denoted by
\begin{equation}\label{def_set}
\mathcal{C}_k(\bm X) = \bigg\{ p_{V|\bm X} \:\bigg| \:V\independent \bm X^{\alpha},\forall \alpha\subseteq [n], |\alpha|=k \bigg\}.
\end{equation}
A variable $V$ generated via a $k$-synergistic channel is said to be a
$k$-synergistic observable.

With this definition, we can consider the $k$\textsuperscript{th}-order synergy
to be the maximum information extractable from a $k$-synergistic channel:
\begin{equation}\label{eq:def_syn}
\texttt{Syn}^{(k)}_\star(\bX; \bY) \coloneqq \sup_{\substack{p_{V|\bX}\in \mathcal{C}_k(\bX):\\V-\bX-\bY}} I(V;\bY)~.
\end{equation} 
This idea can be naturally extended to the case of causal decoupling by
requiring synergistic channels at both sides, i.e.
\begin{equation}\label{eq:def_syn}
\mathcal{G}^{(k)}_\star (\bX; \bY) \coloneqq 
\sup_{\substack{
p_{V|\bX}\in \mathcal{C}_k(\bX),\\
p_{U|\bY}\in \mathcal{C}_k(\bY):\\V-\bX-\bY- U}} I(V;U)~.
\end{equation} 
Finally, the downward causation index can be computed from the difference
\begin{equation}
\mathcal{D}^{(k)}_\star (\bm X_t; \bm X_{t'}) \coloneqq \texttt{Syn}^{(k)}_\star (\bm X_t; \bm X_{t'}) - \mathcal{G}^{(k)}_\star (\bm X_t; \bm X_{t'})~.
\end{equation}
Note that $\texttt{Syn}^{(k)}_\star \geq \mathcal{G}^{(k)}_\star$, which is a
direct consequence of the data processing inequality applied on $V-\bm X_t -
\bm X_{t'} - U$, and therefore $\mathcal{D}_\star^{(k)},
\mathcal{G}_\star^{(k)} \geq 0$.

By exploiting the properties of this specific way of measuring synergy, one can
prove the following result. For this, let us say that a feature $V_t$ is
auto-correlated if $I(V_t;V_{t'}) > 0$.
\begin{proposition}\label{prop:autoemergence}
If \bX is stationary, all auto-correlated $k$-synergistic observables are
$k$\textsuperscript{th}-order emergent.
\end{proposition}
\begin{proof}
See Appendix~\ref{app:proofs2}.
\end{proof}

In summary, $\mathcal{D}^{(k)}_\star$ and $\mathcal{G}^{(k)}_\star$ provide
data-driven tools to test -- and possibly reject -- hypotheses about emergence
in scenarios of interest. Efficient algorithms to compute these quantities are
discussed in Ref.~\cite{rassouli2019data}. Although current implementations
allow only relatively small systems, this line of thinking shows that future
advances in PID might make the computation of emergence indices more scalable,
avoiding the limitations of Eqs.~\eqref{eq:emergence_psi}.

\section{Case studies}
\label{sec:case_studies}

Let us summarise our results so far. We began by formulating a rigorous
definition of emergent features based on PID (Section~\ref{sec:defining}), and
then used $\Phi$ID to break down the emergence capacity into the causal
decoupling and downward causation indices (Section~\ref {sec:taxonomy}).
Although these are not straightforward to compute, the $\Phi$ID framework
allows us to formulate readily computable sufficiency conditions
(Section~\ref{sec:general_tools}). This section illustrates the usage of those
conditions in various case studies.

\subsection{Canonical examples of putative emergence}

Here we present an evaluation of our practical criteria for emergence
(Proposition~\ref{prop:emergence_psi}) in two well-known systems: Conway's Game
of Life (GoL)~\cite{conway1970game}, and Reynolds' flocking boids
model~\cite{reynolds1987flocks}. Both are widely regarded as paradigmatic
examples of emergent behaviour, and have been thoughtfully studied in the
complexity and artificial life literature~\cite{adamatzky2012collision}.
Accordingly, we use these models as test cases for our methods. Technical
details of the simulations are provided in Appendix~\ref{app:sims}.

\subsubsection{Conway's Game of Life}

A well-known feature of GoL is the presence of \emph{particles}: coherent,
self-sustaining structures known to be responsible for information transfer and
modification~\cite{lizier2010local}. These particles have been the object of
extensive study, and detailed taxonomies and classifications
exist~\cite{adamatzky2012collision,wolfram2002new}.

To test the emergent properties of particles, we simulate the evolution of
15x15 square cell arrays, which we regard as a binary vector $\bX \in
\{0,1\}^n$ with $n=225$. As initial condition, we consider configurations that
correspond to a ``particle collider'' setting, with two particles of known type
facing each other (Figure~\ref{fig:gol}). In each trial, the system is
randomised by changing the position, type, and relative displacement of the
particles. After an intial configuration has been selected, the well-known GoL
evolution rule~\cite{conway1970game} is applied 1000 times, leading to a final
state $\bY$~\footnote{Simulations showed that this interval is enough for the
system to settle in a stable state after the collision.}.
\begin{figure}[t]
  \centering
  \includetikz{tikz/}{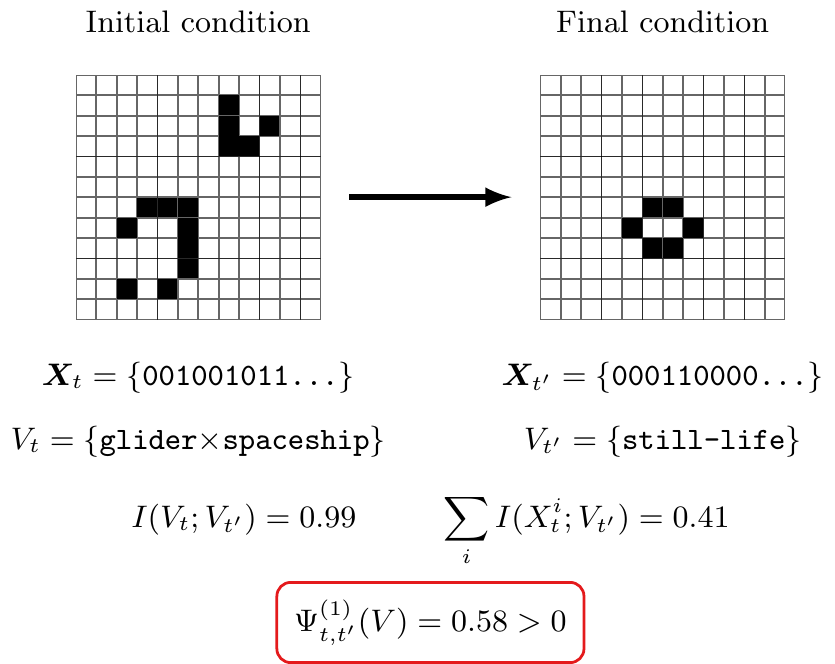}
  \caption{\textbf{Causal emergence in Conway's Game of Life}. The system is
  initialised in a ``particle collider'' setting, and run until a stable
  configuration is reached after the collision. Using particle type as a
  supervenient feature $V$, we find the system meets our practical criterion
  for causal emergence.}

  \label{fig:gol}
\end{figure}

To use the criteria from Eqs.~\ref{eq:emergence_psi}, we need to choose a
candidate emergent feature $V_t$. In this case, we consider a symbolic,
discrete-valued vector that encodes the type of particle(s) present in the
board. Specifically, we consider $\bm V_t = (V_t^1,\dots,V_t^L)$, where $V_t^j
= 1$ iff there is a particle of type $j$ at time $t$ -- regardless of its
position or orientation.

With these variables, we compute the quantities in Eqs.~\ref{eq:emergence_psi}
using Bayesian estimators of mutual information~\cite{archer2013bayesian}. The
result is that, as expected, the criterion for causal emergence is met with
$\Psi^{(1)}_{t,t'}(V) = 0.58 \pm 0.02$. Furthermore, we found that
$\Gamma^{(1)}_{t,t'}(V) = 0.009 \pm 0.0002$, which is orders of magnitude
smaller than $I(V_t; V_{t'}) = 0.99 \pm 0.02$. Using
Proposition~\ref{prop:emergence_psi}, these two results suggest that particle
dynamics in GoL may not only be emergent, but causally decoupled with respect
to their substrate.

\subsubsection{Reynolds' flocking model}

As a second test case, we consider Reynolds' model of flocking behaviour. This
model is composed by \emph{boids} (bird-oid objects), with each boid
represented by three numbers: its position in 2D space and its heading angle.
As candidate feature for emergence, we use the 2D coordinates of the center of
mass of the flock, following Seth~\cite{seth2010measuring}.

In this model boids interact with one another following three rules, each
regulated by a scalar parameter~\cite{seth2010measuring}:
\begin{itemize}[topsep=3pt,itemsep=0pt]
  \item \textbf{aggregation} ($a_1$), as they fly towards the center of the flock;
  \item \textbf{avoidance} ($a_2$), as they fly away from their closest neighbour; and
  \item \textbf{alignment} ($a_3$), as they align their flight direction to that of their neighbours.
\end{itemize}
Following Ref.~\cite{seth2010measuring}, we study small flocks of $N=10$ boids
with different parameter settings to showcase some properties of our practical
criterion of emergence. Note that this study is meant as an illustration of the
proposed theory, and not as a thorough exploration of the flocking model, for
which a vast literature exists (see e.g. the work of
Vicsek~\cite{vicsek2008universal} and references therein).

Figure~\ref{fig:boids} shows the results of a parameter sweep over the
avoidance parameter, $a_2$, while keeping $a_1$ and $a_3$ fixed. When there is
no avoidance, boids orbit around a slowly-moving center of mass, in what could
be called an ordered regime. Conversely, for high values of $a_2$ neighbour
repulsion is too strong for lasting flocks to form, and isolated boids spread
across the space avoiding one another. For intermediate values, the center of
mass traces a smooth trajectory, as flocks form and disintegrate. In line
with the findings of Seth~\cite{seth2010measuring}, our criterion indicates
that the flock exhibits causally emergent behaviour in this intermediate range.

\begin{figure}[t]
  \centering
  \includetikz{tikz/}{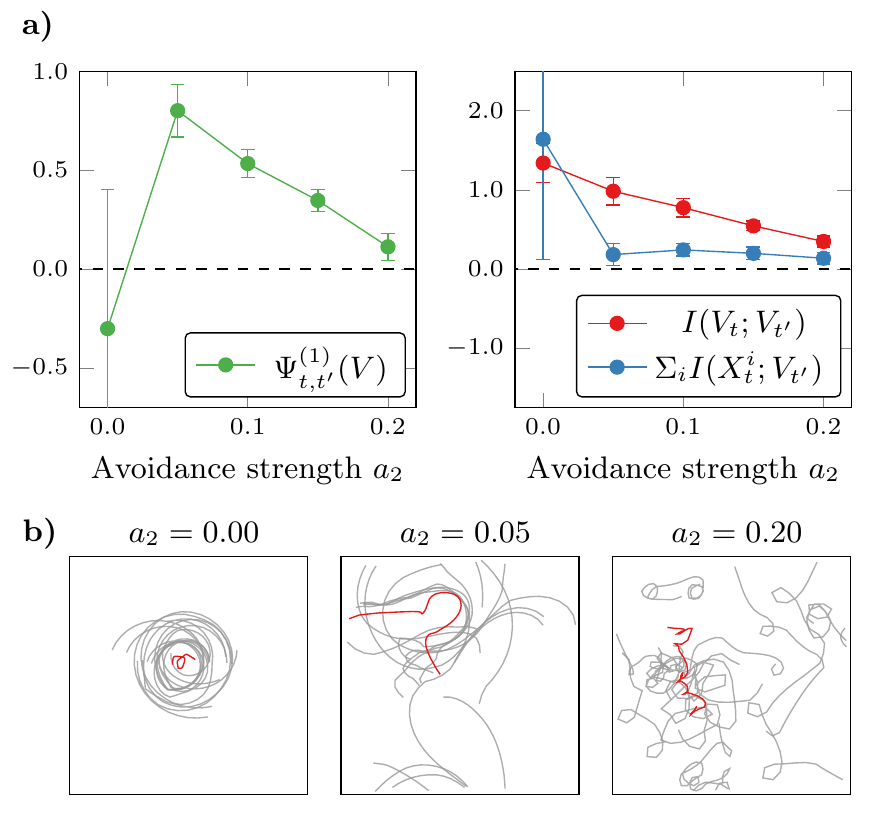}

  \caption{\textbf{Causal emergence in the flocking boids model}. As the
    avoidance parameter is increased, the flock transitions from an attractive
    regime (in which all boids orbit regularly around a stable center of mass),
    to a repulsive one (in which boids spread across space and no flocking is
visible). \textbf{a)} Our criterion $\Psi$ detects causal emergence in an
intermediate range of the avoidance parameter. \textbf{b)} Sample trajectories
of boids (grey) and their center of mass (\textcolor{BrewerRed}{red}).}

  \label{fig:boids}
\end{figure}

\begin{figure*}[t]
  \centering
  \includetikz{tikz/}{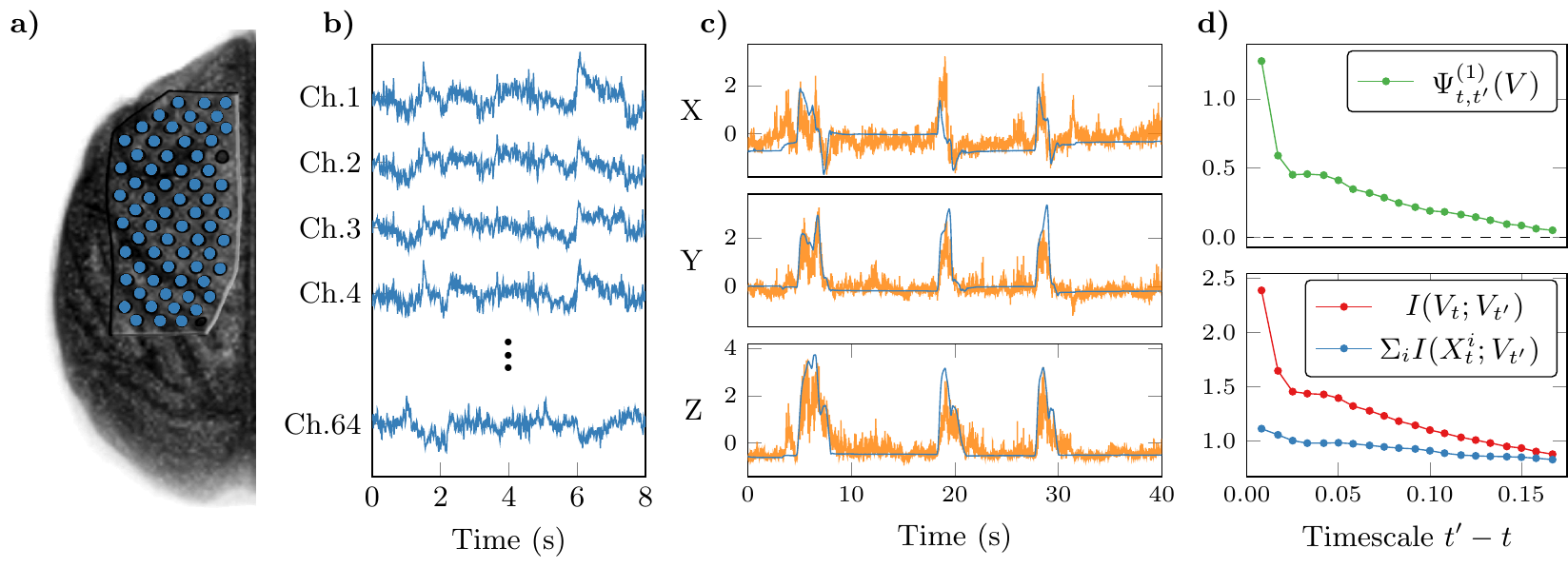}

  \caption{\textbf{Causal emergence in motor behaviour of an awake macaque
monkey}.  \textbf{a)} Position of electrocorticogram (ECoG) electrodes used in
the recording (in blue) overlaid on an image of the macaque's left hemisphere
(front of the brain towards the top of the page). \textbf{b)} Sample time
series from the 64-channel ECoG recordings used, which correspond to the system
of interest $\bX \in \mathbb{R}^{64}$.  \textbf{c)} 3D position of the
macaques's wrist, as measured by motion capture (\textcolor{BrewerBlue}{blue})
and as predicted by the regression model (\textcolor{BrewerOrange}{orange}),
taken as a supervenient feature $V_t \in \mathbb{R}^{3}$. \textbf{d)} Our
emergence criterion yields $\Psi^{(1)}_{t,t'}(V) > 0$, detecting causal
emergence of the behaviour with respect to the ECoG sources.  Original data and
image from Ref.~\cite{chao2010long} and the
\href{www.neurotycho.org}{Neurotycho} database.}

  \label{fig:monkey}
\end{figure*}

By studying separately the two terms that make up $\Psi$ we found that the
criterion of emergence fails for both low and high $a_2$, but for different
reasons (see Figure~\ref{fig:boids}). In effect, for high $a_2$ the
self-predictability of the center of mass (i.e. $I(V_t;V_{t'})$) is low; while
for low $a_2$ it is high, yet lower than the mutual information from individual
boids (i.e. $\sum_i^n I(X_t^i; V_{t'})$). These results suggest that the
low-avoidance scenario is dominated not by a reduction in synergy, but by an
increase in redundancy, which effectively increases the synergy threshold
needed to detect emergence. However, note that, due to the limitations of the
criterion, the fact that $\Psi^{(1)}_{t,t'} < 0$ is inconclusive and does not
rule out the possibility of emergence. This is a common limitation of
whole-minus-sum estimators like $\Psi$; further refinements may provide bounds
that are less susceptible to these issues and perform accurately in these
scenarios.
\vspace{-0.2cm}

\subsection{Mind from matter:\\Emergence, behaviour, and neural dynamics}
\label{sec:mind_from_matter}

A tantalising outcome of having a formal theory of emergence is the capability
of bringing a quantitative angle to the archetype of emergence: the mind-matter
relationship~\cite{dehaene2014consciousness,turkheimer2019conflicting}. 
As a first step in this direction, we conclude this section with an application of 
our emergence criteria to neurophysiological data.

We study simultaneous electrocorticogram (ECoG) and motion capture (MoCap) data
of Japanese macaques performing a reaching task~\cite{chao2010long}, obtained
from the online \href{www.neurotycho.org}{Neurotycho} database. Note that the
MoCap data cannot be assumed to be a supervenient feature of the available ECoG
data, since it doesn't satisfy the conditional independence conditions required
by our definition of supervenience (see Sec.~\ref{sec:formal_theory})
~\footnote{This is likely to be the case because the neural system is only
partially observed -- i.e. the ECoG does not capture every source of relevant
activity in the macaque's cortex. Note that non-supervenient features are of
limited interest within our framework, as they can satisfy
Proposition~\ref{prop:emergence_psi} in trivial ways (e.g. time series which
are independent of the underlying system satisty $\Psi > 0$ if they are
self-correlated).}. Instead, we focus on the portion of neural activity encoded
in the ECoG signal that is relevant to predict the macaque's behaviour, and
conjecture this information to be an emergent feature of the underlying neural
activity (Fig.~\ref{fig:monkey}).

To test this hypothesis, we take the neural activity (as measured by 64 ECoG
channels distributed across the left hemisphere) to be the system of interest,
and consider a memoryless predictor of the 3D coordinates of the macaque's
right wrist based on the ECoG signal. Therefore, in this scenario $\bm
X_t\in\mathbb{R}^{64}$ and $V_t = F(\bm X_t) \in \mathbb{R}^3$. To build $V_t$,
we used Partial Least Squares (PLS) and a Support Vector Machine (SVM)
regressor, the details of which can be found in Appendix~\ref{app:ecog}.

After training the decoder and evaluating on a held-out test set, results show
that $\Psi > 0$, confirming our conjecture that the motor-related information
is an emergent feature of the macaque's cortical activity. For short timescales
($t'-t = 8\,\mathrm{ms}$), we find $\Gamma_{t,t'}^{(1)}(V) = 0.049 \pm 0.002$,
which is orders of magnitude smaller than $\Psi_{t,t'}^{(1)}(V) = 1.275 \pm
0.002$, suggesting that the behaviour may have an important component decoupled
from individual ECoG channels. Furthermore, the emergence criterion is met for
multiple timescales $t'-t$ of up to $\approx 0.2\,\mathrm{s}$, beyond which the
predictive power in $V_t$ and individual electrodes decrease and become nearly
identical.

This analysis, while just a proof of concept, helps us quantify how and to what
extent behaviour emerges from collective neural activity; and opens the door to
further tests and quantitative empirical explorations of the mind-matter
relationship.

\section{Discussion}
\label{sec:discussion}

A large fraction of the modern scientific literature considers strong emergence
to be impossible or ill-defined. This judgement is not fully unfounded: a
property that is simultaneously supervenient (i.e. that can be computed from
the state of the system) and that has irreducible causal power (i.e. that
``tells us something'' that the parts don't) can indeed seem to be an
oxymoron~\cite{bedau1997weak}. Nonetheless, by linking supervenience to static
and causal power to dynamical properties, our framework shows that these two
phenomena are perfectly compatible within the -- admittedly counterintuitive --
laws of multivariate information dynamics~\cite{mediano2019beyond}, providing a
tentative solution to this paradox. 

Our theory of causal emergence is about predictive power, not
``explicability''~\cite{chalmers2006strong}, and therefore is not related to
views on strong emergence such as Chalmers'~\cite{chalmers2006strong}.
Nevertheless, our framework embraces aspects that are commonly associated with
strong emergence -- such as downward causation -- and renders them
quantifiable. Our framework also does not satisfy conventional definitions of
weak emergence~\footnote{Note that the systems studied in
Section~\ref{sec:examples} are not weakly emergent in the sense of
Bedau~\cite{bedau1997weak}, being simple and susceptible to explanatory
shortcuts.}, but is compatible with more general notions of weak emergence,
e.g. the one introduced by Seth (see Sec.~\ref{sec:hoel}). Hence, our theory
can be seen as an attempt at reconciling these approaches~\cite{turkheimer2019conflicting}, 
showing how ``strong'' a ``weak'' framework can be.

An important consequence of our theory is the fundamental connection
established between causal emergence and statistical synergy: the system's
capacity to host emergent features was found to be determined by how
synergistic its elements are with respect to their future evolution. Although
previous ideas about synergy have been loosely linked to emergence in the
past~\cite{corning1998synergism}, this is (to the best of our knowledge) the
first time such ideas have been formally laid out and quantified using recent
advances in multivariate information theory.

Next, we examine a few caveats regarding the applicability of the proposed
theory, its relation with prior work, and some open problems.

\subsection{Scope of the theory}

Our theory focuses on \emph{synchronic}~\cite{rueger2000physical} aspects of
emergence, analysing the interactions between the elements of dynamical systems
and collective properties of them as they jointly evolve over time. As such,
our theory directly applies to any system with well-defined dynamics, including
systems described by deterministic dynamical systems with random initial
conditions~\cite{rosas2018selforg} and stochastic systems described by
Fokker-Planck equations~\cite{breuer2002theory}. In contrast, the application
of our theory to systems in thermodynamic equilibrium may not be
straightforward, as their dynamics are often not uniquely specified by the
corresponding Gibbs distributions~\footnote{ For an explicit example, when
considering the Ising model, Kawasaki and Glauber dynamics are known to behave
differently even when the system is in equilibrium~\cite{smith2008interfaces};
and thus may provide quite different values of the measures described
in~\ref{sec:measuring}.}. Finding principled approaches to guide the
application of our theory to those cases is an interesting challenge for future
studies.

In addition, given the breadth of the concept of ``emergence,'' there are a
number of other theories leaning more towards philosophy that are orthogonal to
our framework. This includes, for example, theories of emergence as
\textit{radical novelty} (in the sense of features not previously observed in
the system)~\cite{corning2002reemergence}, most prominently encapsulated in the
aphorism ``more is different'' by Anderson~\footnote{See
Refs.~\cite{anderson1972more,anderson2018basic}, particularly his approach to
emergence in biology. Note that some of Anderson's views -- particularly the
ones related to rigidity -- are nevertheless closely related to the approach
developed by our framework.}, and also articulated in the work of
Kauffman~\cite{kauffman2006emergence,kauffman2019world}. Also,
\textit{contextual emergence} emphasises a role for macro-level contexts that
cannot be described at the micro-level, but which impose constraints on the
micro-level for the emergence of the
macro~\cite{bishop2006contextual,atmanspacher2009contextual}. These are
valuable philosophical positions, which have been studied from a statistical
mechanics perspective in
Ref.~\cite{jensen2018statistical,bishop2006contextual}. Future work shall
attempt to unify these other approaches with our proposed framework.

\subsection{Causality}

The \textit{de facto} way to assess the causal structure of a system is to
analyse its response to controlled interventions or to build intervention
models (causal graphs) based on expert knowledge, which leads to the well-known
\textit{do-calculus} spearheaded by Judea Pearl~\cite{pearl2000causality}. This
approach is, unfortunately, not applicable in many scenarios of interest, as
interventions may incur prohibitive costs or even be impossible, and expert
knowledge may not be available. These scenarios can still be assessed via the
Wiener-Granger theory of \textit{statistical causation}, which studies the
blueprint of predictive power across the system of interest by accounting
non-mediated correlations between past and future
events~\cite{bressler2011wiener}. Both frameworks provide similar results when
all the relevant variables have been measured, but can neverthelss differ
radically when there are unobserved interacting
variables~\cite{pearl2000causality}. The debate between the Wiener-Granger and
the Pearl schools has been discussed in other related contexts -- see e.g.
Refs.~\cite{mediano2018measuring,barrett2019phi} for a discussion regarding
Integrated Information Theory (IIT), and Ref.~\cite{seth2015granger} for a
discussion about effective and functional connectivity in the context of
neuroimaging time series analysis~\footnote{ In a nutshel, effective
connectivity aims to uncover the minimal physical causal mechanism underlying
the observed data, while functional connectivity describes directed or
undirected statistical dependences~\cite{seth2015granger}.}.

In our theory, the main object of analysis is Shannon's mutual information,
$I(\bX; \bY)$, which depends on the joint probability distribution $p_{\bX,
\bY}$. The origin of this distribution (whether it was obtained by passive
observation or by active intervention) will change the interpretation of the
quantities presented above, and will speak differently to the Pearl and the
Wiener-Granger schools of thought; some of the implications of these
differences are addressed when discussing Ref.~\cite{hoel2013quantifying}
below. Nonetheless, since both methods of obtaining $p_{\bX, \bY}$ allow
synergy to take place, our results are in principle applicable in both
frameworks -- which allows us to formulate our theory of causal emergence
without taking a rigid stance on a theory of causality itself.

\subsection{Relationship with other quantitative theories of emergence}
\label{sec:hoel}

This work is part of a broader movement towards formalising theories of
complexity through information theory. In particular, our framework is most
directly inspired by the work of Seth~\cite{seth2010measuring} and Hoel
\emph{et al.}~\cite{hoel2013quantifying}, and also related to recent work by
Chang \emph{et al.}~\cite{chang2019information}. This section gives a brief
account of these theories, and discusses how they differ from our proposal.

Seth~\cite{seth2010measuring} proposes that a process $V_t$ is Granger-emergent
(or \textit{G-emergent}) with respect to $\bX$ if two conditions are met: (i)
$V_t$ is autonomous with respect to $\bX$ (i.e. $I(V_t; V_{t'} | \bX)>0$), and
(ii) $V_t$ is G-caused by $\bX$ (i.e. $I(\bX; V_{t'} | V_t)>0$). The latter
condition is employed to guarantee a relationship between $\bX$ and $V_t$; in
our framework an equivalent role is taken by the requirement of supervenience.
The condition of autonomy is certainly related with our notion of causal
decoupling. However, as shown in Ref.~\cite{williams2010nonnegative}, the
conditional mutual information conflates unique and synergistic information,
which can give rise to undesirable situations: for example, it could be that
$I(V_t; V_{t'} | \bX) >0$ while, at the same time, $I(V; V_{t'}) = 0$, meaning
that the dynamics of the feature $V_t$ are only visible when considering it
together with the full system \bX, but not on its own. Our framework avoids
this problem by refining this criterion via PID, and uses only the unique
information for the definition of emergence.

Our work is also strongly influenced by the framework put forward by Hoel and
colleagues in Ref.~\cite{hoel2013quantifying}. Their approach is based on a
coarse-graining function $F(\cdot)$ relating a feature of interest to the
system, $V_t = F(\bX)$, which is a particular case of our more general
definition of supervenience. Emergence is declared when the dependency between
$V_t$ and $V_{t'}$ is ``stronger'' than the one between $\bX$ and $\bY$. Note
that $V_t - \bX - \bY - V_{t'}$ is a Markov chain, and hence $I(V_t;V_{t'})
\leq I(\bX;\bY)$ due to the data processing inequality; therefore, a direct
usage of Shannon's mutual information would make the above criterion impossible
to fulfill. Instead, this framework focuses on the transition probabilities
$p_{V_{t'}|V_t}$ and $p_{\bm X_{t'}|\bm X_t}$, and hence the mutual information
terms are computed using maximum entropy distributions instead of the
stationary marginals. By doing this, Hoel \textit{et al.} account not for what
the system \emph{actually does}, but for all the potential transitions the
system \emph{could do}. However, in our view this approach is not well-suited
to assess dynamical systems, as it might account for transitions that are never
actually explored~\footnote{The difference between stationary and maximum
entropy distributions can be particularly dramatic in non-ergodic systems with
multiple attractors. For a related discussion in the context of Integrated
Information Theory, see Ref.~\cite{barrett2019phi}.}. Additionally, this
framework relies on having exact knowledge about the microscopic transitions as
encoded by $p_{\bm X_{t'}|\bm X_t}$, which is not possible to obtain in most
applications.

Finally, Chang \emph{et al.}~\cite{chang2019information} consider supervenient
variables that are ``non-trivially informationally closed'' (NTIC) to their
corresponding microscopic substrate. NTIC is based on a division of $\bX$ into
a subsystem of interest, $\bX^{\alpha}$, and its ``environment'' given by
$\bX^{-\alpha}$. Interestingly, a system being NTIC requires $V_t$ to be
supervenient only with respect to $\bX^{\alpha}$ (i.e. $V_t=F(\bX^{\alpha})$),
as well as information flow from the environment to the feature (i.e.
$I(\bX^{-\alpha};V_{t'}) > 0$) mediated by the feature itself, so that
$\bX-V_t-V_{t'}$ is a Markov chain. Hence, NTIC requires features that are
sufficient statistics for their own dynamics, which is akin to our concept of
causal decoupling but focused on the interaction between a macroscopic feature,
an agent, and its environment. Extending our framework to agent-environment
systems involved in active inference is part of our future work.

\subsection{Limitations and open problems}
\label{sec:limitations}

The framework presented in this paper focuses on features from fully observable
systems with Markovian dynamics. These assumptions, however, often do not hold
when dealing with experimental data -- especially in biological and social
systems. As an important extension, future work should investigate the effect
of unobserved variables on our measures. This could be done, for example,
leveraging Taken's embedding
theorem~\cite{takens1981detecting,cliff2016information} or other
methods~\cite{wilting2018inferring}.

An interesting feature of our framework is that, although it depends on the
choice of PID and $\Phi$ID, its practical application via the criteria
discussed in Section~\ref{sec:general_tools} is agnostic to those choices.
However, they incur the cost of a limited sensitivity to detect emergence due
to an overestimation of the microscopic redundancy; so they can detect
emergence when it is substantial, but might miss it in more subtle cases.
Additionally, these criteria are unable to rule out emergence, as they are
sufficient but not necessary conditions. Therefore, another avenue of future
work should search for improved practical criteria for detecting emergence from
data. One interesting line of research is providing scalable approximations for
$\texttt{Syn}_\star^{(k)}$ and $\mathcal{G}_\star^{(k)}$ as introduced in
Section~\ref{sec:quantifying_synergies}, which could be computed in large
systems.

Another open question is how the emergence capacity is affected by changes in
the microscopic partition of the system (c.f. Section~\ref{sec:defining}).
Interesting applications of this includes scenarios where elements of interest
have been subject to a mixing process, such as the case of
electroencephalography where each electrode detects a mixture of brain sources.
Other interesting questions include studying systems with non-zero emergence
capacity for all reasonable microscopic partitions, which may correspond to a
stronger type of emergence.

\section{Conclusion}

This paper introduces a quantitative definition of causal emergence, which
addresses the apparent paradox of supervenient macroscopic features with
irreducible causal power using principles of multivariate statistics. We
provide a formal, quantitative theory that embodies many of the
principles attributed to strong emergence, while being measurable and
compatible with the established scientific worldview. Perhaps the most
important contribution of this work is to bring the discussion of emergence
closer to the realm of quantitative, empirical scientific investigation,
complementing the ongoing philosophical inquiries on the subject.

Mathematically, the theory is based on the Partial Information Decomposition
(PID) framework~\cite{williams2010nonnegative}, and on a recent extension,
Integrated Information Decomposition ($\Phi$ID)~\cite{mediano2019beyond}. The
theory allows the derivation of sufficiency criteria for the detection of
emergence that are scalable, easy to compute from data, and based only on
Shannon's mutual information. We illustrated the use of these practical
criteria in three case studies, and concluded that: i) particle collisions are
an emergent feature in Conway's Game of Life, ii) flock dynamics are an
emergent feature of simulated birds; and iii) the representation of motor
behaviour in the cortex is emergent from neural activity. Our theory, together
with these practical criteria, enables novel data-driven tools for
scientifically addressing conjectures about emergence in a wide range of
systems of interest.

Our original aim in developing this theory, beyond the contribution to
complexity theory, is to help bridge the gap between the mental and the
physical, and ultimately understand how mind emerges from matter. This paper
provides formal principles to explore the idea that psychological phenomena
could emerge from collective neural patterns, and interact with each other
dynamically in a causally decoupled fashion -- perhaps akin to the
``statistical ghosts'' mentioned in Section~\ref{sec:decoupling}. Put simply:
just as particles in the Game of Life have their own collision rules, we wonder
if thought patterns could have their own emergent dynamical laws, operating at
a larger scale with respect to their underlying neural
substrate~\footnote{Similar ideas have been recently explored by
Kent~\cite{kent2019toy}.}. Importantly, the theory presented in this paper not
only provides conceptual tools to frame this conjecture rigorously, but also
provides practical tools to test it from data. The exploration of this
conjecture is left as an exciting avenue for future research.

\section*{Acknowledgements}

The authors thank Shamil Chandaria, Matthew Crosby, Martin Biehl and Acer Chang
for insightful discussions, and the creators of the Neurotycho database for
opening to the public such a valuable resource. P.M. and D.B. are funded by the
Wellcome Trust (grant no. 210920/Z/18/Z). F.R. is supported by the Ad Astra
Chandaria foundation. A.K.S. and A.B.B. are grateful to the Dr. Mortimer and
Theresa Sackler Foundation, which supports the Sackler Centre for Consciousness
Science.

\appendix

\section{Information decomposition in large multivariate systems}
\label{app:pid}

\begin{figure*}[t]
  \centering
  \includetikz{tikz/}{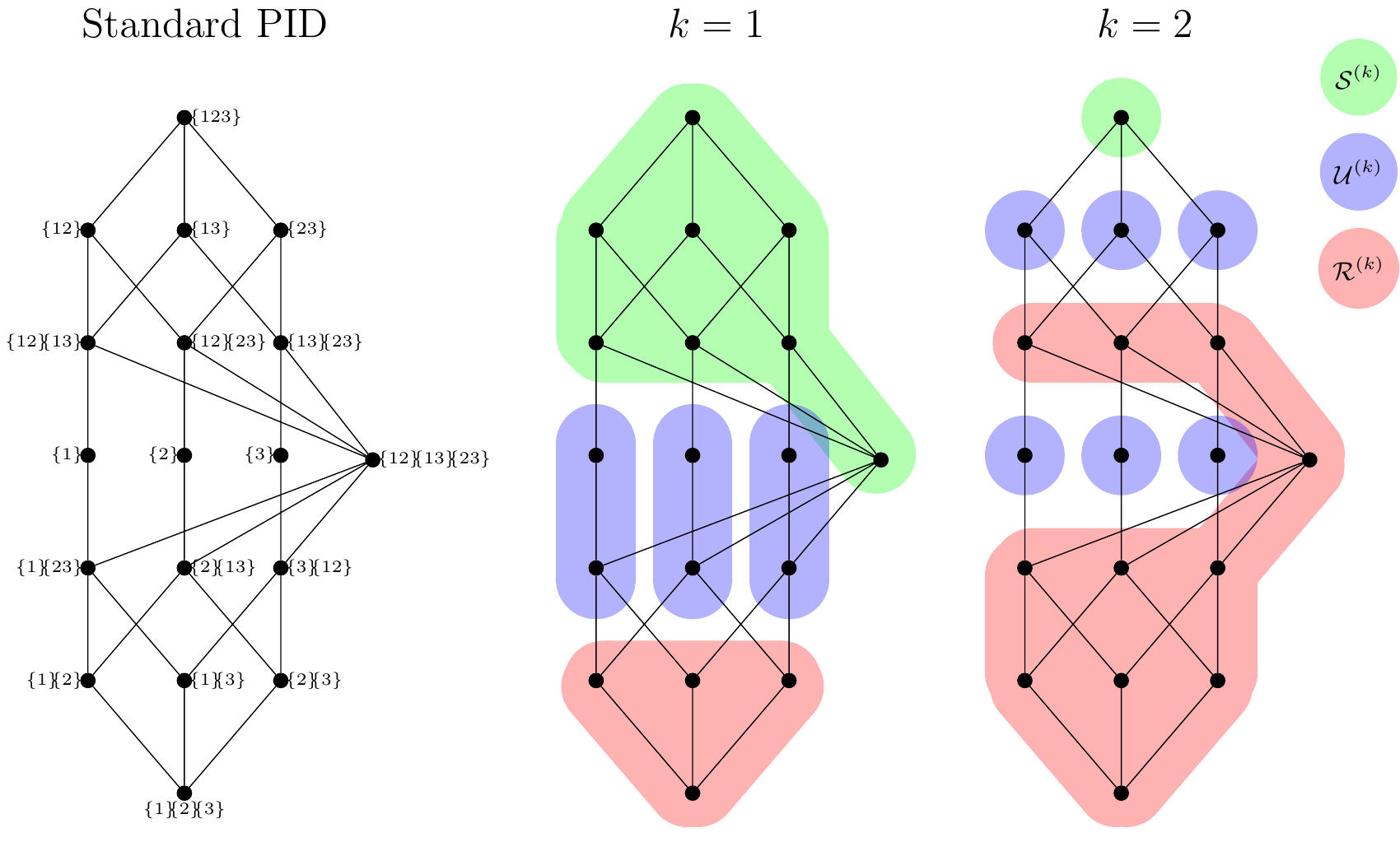}

  \caption{\textbf{Coarse-grained partial information decomposition of order
  $k$}. (\emph{left}) Standard PID lattice for $n=3$, shown for reference.
Node labels are omitted from the other lattices for clarity. \emph{(middle)}
Coarse-graining for $k=1$, superimposed on the PID lattice. (\emph{right})
Coarse-graining for $k=2$. For both values of $k$, Lemma~\ref{lemma:PIDk}
guarantees that the $k$\textsuperscript{th}-order atoms provide an exact
decomposition of mutual information.}

  \label{fig:PIDk}
\end{figure*}

To formulate our theory of causal emergence for arbitrary order $k$, in
Sec.~\ref{sec:pid} we introduced definitions of $k$\textsuperscript{th}-order
synergy and unique information. In this appendix we complement these with a
matching definition of $k$\textsuperscript{th}-order redundancy, and show that
these provide a full-fledged information decomposition for any $k =
\{1,\dots,n-1\}$. For completeness, we present all definitions and examples
here -- including those that were necessary for the exposition of the main text
and were previously presented in Sec.~\ref{sec:pid}.

We begin by (re-)introducing the notion of $k$\textsuperscript{th}-order
synergy between $n$ variables, defined as
\begin{equation}
  \textnormal{\texttt{Syn}}^{(k)}(\bm X ; Y)\coloneqq \sum_{\bm\alpha\in\mathcal{S}^{(k)}} I_\partial^{\bm\alpha}(\bm X ; Y)~,\nonumber
\end{equation}
with $\mathcal{S}^{(k)} = \{\{\alpha_1,\dots,\alpha_L\} \in\mathcal{A} :
|\alpha_j| > k, \forall j=1,\dots,L \}$. Intuitively, $\texttt{Syn}^{(k)}(\bm
X; Y)$ corresponds to the information about the target that is provided by the
whole $\bm X$ but is not contained in any set of $k$ or less parts when
considered separately from the rest. Accordingly, $\mathcal{S}^{(k)}$ only
contains collections of more than $k$ sources. For example, for $n=2$ we obtain
the standard synergy $\mathcal{S}^{(1)} = \{ \{12\} \}$, and for $n=3$ we have
$\splitatcommas{\mathcal{S}^{(1)} = \{ \{12\}, \{13\}, \{23\}, \{12\}\{13\},
\{12\}\{23\}, \{13\}\{23\}, \{12\}\{13\}\{23\}, \{123\}\}}$.

Similarly, the $k$\textsuperscript{th}-order unique information of $\bm
X^\beta$ with $\beta \subset [n]$ is calculated as
\begin{equation}
  \textnormal{\texttt{Un}}^{(k)}(\bm X^{\beta} ; Y | \bm X^{-\beta} ) \coloneqq \sum_{\bm \alpha \in \mathcal{U}^{(k)}(\beta)} I_\partial^{ \bm\alpha}(\bm X ; Y)~, \nonumber
\end{equation}
with $\mathcal{U}^{(k)}(\beta) = \{ \bm \alpha\in \mathcal{A} : \beta \in
\bm\alpha, \forall \alpha\in\bm\alpha \setminus \beta, \alpha\subseteq
[n]\setminus\beta, |\alpha| > k\}$, and $\bm X^{-\beta}$ being all the
variables in $\bm X$ the indices of which are not in $\beta$. This corresponds
to all the atoms where $\beta$ is the only source of size $k$ or less -- which,
importantly, is in general not just $I_\partial^{\beta}$. Intuitively, this is
the information that $\bm X^{\beta}$ has access to and no other subset of
parts has access to \emph{on its own} (although bigger groups of other
parts may). And again, for $n=2$ we recover $\mathcal{U}^{(1)}(\{i\}) = \{
\{i\} \}$, and for $n=3$ we have e.g. $\mathcal{U}^{(1)}(\{1\}) = \{ \{1\},
\{1\}\{23\} \}$.

Finally, the $k$\textsuperscript{th}-order redundancy is given by
\begin{equation}
  \textnormal{\texttt{Red}}^{(k)}(\bm X ; Y ) \coloneqq  \sum_{\bm \alpha \in \mathcal{R}^{(k)}} I_\partial^{ \bm\alpha}(\bm X ; Y)~, \nonumber
\end{equation}
with $\mathcal{R}^{(k)} = \{ \bm\alpha\in \mathcal{A} : \exists i\neq j,
|\alpha_i|, |\alpha_j| \leq k \}$. Intuitively, $\texttt{Red}^{(k)}(\bm X ; Y )$
is the information that is held by at least two different groups of size $k$
or less. Again, in the $n=2$ case we recover the standard redundancy
$\splitatcommas{\mathcal{R}^{(1)} = \{ \{1\}\{2\} \}}$; and as an example for
$n=3$ we have $\splitatcommas{\mathcal{R}^{(1)} = \{ \{1\}\{2\}, \{1\}\{3\},
\{2\}\{3\}, \{1\}\{2\}\{3\} \}}$.

With the definitions above, we can build a coarse-grained PID which generalises
the well-known construction for $n=2$. This allows us to formulate
decompositions with a small number of atoms that scale gracefully with system
size, and, more interestingly, preserve the intuitive meaning that synergy,
redundancy, and unique information have for $n=2$.

\begin{lemma}\label{lemma:PIDk}
The $k$\textsuperscript{th}-order synergy, redundancy, and unique information
defined above provide an exact decomposition of mutual information:
\begin{align}
I(\bm X;Y) =&~ \textnormal{\texttt{Red}}^{(k)}(\bm X;Y) + \textnormal{\texttt{Syn}}^{(k)}(\bm X;Y)  \nonumber \\
&+ \sum_{\substack{\beta \subset [n]: \\|\beta|\leq k}} \textnormal{\texttt{Un}}^{(k)}(\bm X^\beta;Y|\bm X^{-\beta})~.
\end{align}
\end{lemma}
\begin{proof}

We will prove this by showing that the sets $\mathcal{R}^{(k)}$,
$\mathcal{S}^{(k)}$ and $\mathcal{U}^{(k)}(\beta)$ are a partition of
$\mathcal{A}$. We will do this in two steps: first, we show that their
intersection is empty; and second, that their union is $\mathcal{A}$.

Let us show that the intersections between every pair of sets is empty:
\begin{itemize}

  \item $\Rk \bigcap \Sk = \emptyset$, since if $\bm\alpha \in \Rk$ it must
  contain at least one $\alpha \in \bm\alpha : |\alpha| \leq k$, and therefore
  $\bm\alpha \notin \Sk$.

  \item $\Uk(\gamma) \bigcap \Uk(\beta) = \emptyset$ if and only if $\gamma
  \neq \beta$, since every $\bm\alpha \in \Uk(\gamma)$ has either no other
  elements apart from $\gamma$ (in which case $\beta \notin \bm\alpha$ and thus
  $\bm\alpha \notin \Uk(\beta)$), or other elements of cardinality greater than
  $k$ (in which case, again, $\beta \notin \bm\alpha$ and thus $\bm\alpha
  \notin \Uk(\beta)$).

  \item $\Sk \bigcap \Uk(\beta) = \emptyset$ for all $|\beta| \leq k$, since
  every $\bm\alpha \in \Uk(\beta)$ contains at least one element with
  cardinality less than or equal to $k$ (specifically, $\beta$), and therefore
  $\bm\alpha \notin \Sk$.

  \item $\Rk \bigcap \Uk(\beta) = \emptyset$ for all $|\alpha| \leq k$, since
  every $\bm\alpha \in \Rk$ contains at least two sets with cardinality less
  than or equal to $k$, while by the definition of $\Uk(\beta)$ every element
  other than $\beta$ must have cardinality greater than $k$, and thus
  $\bm\alpha \notin \Uk(\beta)$.

\end{itemize}

This concludes the first part of the proof. Next, we need to prove that the
union of those sets is indeed $\mathcal{A}$. We will show this by proving that
every $\bm\alpha \in \mathcal{A}$ is in one of those sets:

\begin{itemize}

  \item If $\nexists \alpha \in \bm\alpha : |\alpha| \leq k$, then $\bm\alpha
  \in \Sk$.

  \item If there is exactly one $\alpha \in \bm\alpha : |\alpha| \leq k$, then
  $\bm\alpha \in \Uk(\alpha)$.

  \item If there are more than one $\alpha \in \bm\alpha : |\alpha| \leq k$,
  then $\bm\alpha \in \Rk$.

\end{itemize}

\end{proof}

The two possible decompositions for $n=3$, together with the standard PID
lattice, are shown in Figure~\ref{fig:PIDk}.

\section{Properties of high-order PI atoms}
\label{app:highorderPID}

As discussed in Section~\ref{sec:formal_theory}, our theory does not depend on
a specific functional form of PID. Instead, it applies to any PID that
satisfies the following properties:

\begin{itemize}

  \item \emph{Deterministic equality}: if there exists a function $f(\cdot)$ 
    such that $X^i = f(X^j)$ with $i\neq j$, then the information decomposition of $\bm X$ 
    is isomorphic to that of $\bm X^{-j}$ (see below).

  \item \emph{Non-negativity}: $\textnormal{\texttt{Syn}}^{(k)}( \bm X; Y) \geq
  0$, and $\min\{I(\bm Z;Y), I(\bm Z;Y | \bm X) \} \geq \textnormal{\texttt{Un}}^{(k)}(
  \bm Z; Y | \bm X) \geq 0$.

  \item \textit{Source data processing inequality}: $\textnormal{\texttt{Un}}^{(k)}( \bm W; Y | \bm X) \leq
  \textnormal{\texttt{Un}}^{(k)}( \bm Z; Y | \bm X)$ for all $\bm W-\bm Z-(\bm X,Y)$
  Markov chains.

\end{itemize}

To formulate the causal decoupling and downward causation indices in
Section~\ref{sec:taxonomy}, we make use of $\Phi$ID, a recent extension of PID
to multi-target settings~\cite{mediano2019beyond}. As with PID, our theory does
not require a particular functional form of $\Phi$ID, only the following
property:

\begin{itemize}
  \item $\sum_{|\alpha|=k}\texttt{Syn}^{(k)}(\bX; \bm X_{t'}^{\alpha}) \geq \mathcal{D}^{(k)}(\bX;\bY) \geq \texttt{Syn}^{(k)}(\bX; \bm X_{t'}^{\beta})$ for all $|\beta|=k$.
\end{itemize}

Finally, these properties are required to formulate the practical criteria for
emergence in Section~\ref{sec:general_tools}:

\begin{itemize}

  \item \emph{Whole-minus-sum}: $\textnormal{\texttt{Syn}}^{(k)}( \bm X; Y)
  \geq I(\bm X;Y) - \sum_{|\alpha|=k} I(\bm X^{\alpha} ;Y)$.

  \item \emph{Target data processing inequality}: for all $\bm X-Y-U$ Markov
  chains, $\textnormal{\texttt{Syn}}^{(k)}( \bm X; U) \leq
  \textnormal{\texttt{Syn}}^{(k)}( \bm X; Y)$.

\end{itemize}

For completeness, we present a precise definition of the deterministic equality
property, as previously introduced in the PID
literature~\cite{kolchinsky2019novel}.

\begin{definition}
A PID satisfies \emph{deterministic equality} if 
$I_\partial^{\bm\alpha}(\bm X ; Y) = I_\partial^{g_j(\bm\alpha)}(\bm X^{-j} ; Y)$ for all 
$\bm\alpha\in\mathcal{A}$ 
whenever there exists a function $f(\cdot)$ such that $f(X^j) = X^i$ with $i\neq j$, 
with $g_j(\bm\alpha)$ removes $j$ from all the sets of indices in $\bm\alpha$.
\end{definition}

It is direct to check that a number of well-known information decompositions,
including the Minimum Mutual Information PID~\cite{barrett2015exploration}, and
the corresponding $\Phi$ID~\cite{mediano2019beyond}, satisfy these
requirements.

With these properties at hand we can prove the following results, used in
Sections~\ref{sec:defining} and~\ref{sec:taxonomy}:

\begin{lemma}\label{lemma:unsyn}
If $X^{n+1} = \bm X $, then the following holds:
\begin{equation}
  \textnormal{\texttt{Syn}}^{(k)}(\bm X;Y)
  = \textnormal{\texttt{Un}}^{(k)}( X^{n+1} ;Y | \bm X)~.
\end{equation}
Above, the second term corresponds to a PID over a system of $n+1$ elements.
\end{lemma}
\begin{proof}
Begin by considering the PID of $n$ sources $X^1, ... , X^n$ on the lattice
$\mathcal{A}^n$, and define its set $\mathcal{S}^{(k)}$ as above. Now we add an
additional $n+1^{\mathrm{st}}$ variable that is simply all of them
concatenated, $X_{n+1} = \bm X^n$, and build a PID on the lattice
$\mathcal{A}^{n+1}$. In the following, we consider the set
$\mathcal{S}^{(k)}$ to belong to the lattice $\mathcal{A}^n$, and the set
$\mathcal{U}^{(k)}(\beta)$ to belong to the lattice $\mathcal{A}^{n+1}$. To
prove the lemma we need four ingredients, which we provide in the four
paragraphs below.

\begin{enumerate}

\item First, note that the nodes in $\mathcal{A}^{n+1}$ that precede $\{n+1\}$
are those in $\mathcal{A}^n$, but with the singleton $\{n+1\}$ appended to
them. More specifically, the mapping $f: \mathcal{A}^n \rightarrow
\mathcal{A}^{n+1}$ of the form $f(\bm\alpha) = \bm\alpha\cup \{\{n+1\}\}$ is
such that for any $\bm\alpha \in \mathcal{A}^n$ then $f(\bm\alpha) \prec
\{\{n+1\}\}$. Additionally, due to the properties of the partial order,
$\bm\alpha \preceq \bm\alpha'$ if and only if $f(\bm\alpha) \preceq
f(\bm\alpha')$. This shows that $\mathcal{A}^n$ is isomorphic to a sublattice
of $\mathcal{A}^{n+1}$.

\item Next, by the \emph{deterministic equality} property it is direct to check
that $I_\cap^{f(\bm\alpha)} = I_\cap^{\bm\alpha\vphantom{f}}$. Since this
equality holds for all $\bm\alpha \in \mathcal{A}^n$, then applying a M\"obius
inversion we directly obtain that $I_\partial^{f(\bm\alpha)} =
I_\partial^{\bm\alpha \vphantom{f}}$.

\item Additionally, by construction of $\mathcal{U}^{(k)}$ and
$\mathcal{S}^{(k)}$, for all $\bm\gamma \in \mathcal{S}^{(k)}$ one has
$f(\bm\gamma) \in \mathcal{U}^{(k)}(\{n+1\})$. In other words, the set
$\mathcal{U}^{(k)}(\{n+1\})$ includes all atoms in $\mathcal{S}^{(k)}$, plus
$\{n+1\}$.

\item Finally, note that the node $\{12...n\}\{n+1\}$ is the only direct
predecessor of $\{n+1\}$, since there exists no node $\bm\beta \in
\mathcal{A}^{n+1}$ such that $\bm\beta \prec \{n+1\}$ and not $\bm\beta \preceq
\{12...n\}\{n+1\}$. By the \emph{deterministic equality} property
$I_\cap^{\{12...n\}\{n+1\}} = I_\cap^{\{12...n\}}$ and, therefore,
$I_\partial^{\{n+1\}} = 0$.

\end{enumerate}

With all of these, it is direct to see that
\begin{align*}
  \texttt{Un}^{(k)}(\bm X^n;Y|\bm X) &= \sum_{\bm \alpha \in \mathcal{U}^{(k)}(\{n+1\})} I_\partial^{\bm\alpha} (\bm X; Y) \\
  &= \sum_{\bm \alpha \in \mathcal{S}^{(k)}} I_\partial^{ \bm\alpha} (\bm X; Y) = \texttt{Syn}^{(k)}(\bm X;Y)~.
\end{align*}
\end{proof}

\begin{corollary}\label{cor:unsyn_phiid}
If $X^{n+1} = \bm X$ and $Y^{n+1} = \bm Y$, then the following holds:
\begin{equation}
  \mathcal{G}^{(k)}(\bm X; \bm Y) = \textnormal{\texttt{Un}}^{(k)}(X^{n+1}; Y^{n+1} | \bm X, \bm Y)
\end{equation}
Above, the second term corresponds to a $\Phi$ID over a system of $n+1$
elements.
\end{corollary}
\begin{proof}
Follows from a direct $\Phi$ID extension to the proof of
Lemma~\ref{lemma:unsyn}, by formulating a decomposition
\begin{align*}
  I(\bm X; \bm Y) = \sum_{\bm\alpha,\bm\beta \in \mathcal{A}^n} I_\partial^{\bm\alpha \rightarrow \bm\beta} (\bm X; \bm Y) ~ ,
\end{align*}
and applying the proof above to both $\bm\alpha$ and $\bm\beta$. Strictly
speaking, this also requires a natural multi-target extension of the
Deterministic Equality property, namely that $I_\cap^{\bm\alpha \rightarrow
\bm\beta}(\bm X; \bm Y) = I_\cap^{\bm\alpha \rightarrow \bm\beta}(\bm X^{-j};
\bm Y)$ if $X^i = f(X^j)$ with $j \neq i$, for any $\bm\beta \in
\mathcal{A}^n$; as well as the symmetric $I_\cap^{\bm\alpha \rightarrow
\bm\beta}(\bm X; \bm Y) = I_\cap^{\bm\alpha \rightarrow \bm\beta}(\bm X; \bm
Y^{-j})$ if $Y^i = f(Y^j)$ with $j \neq i$, for any $\bm\alpha \in
\mathcal{A}^n$.
\end{proof}

\section{Mathematical properties of causal emergence}
\label{app:proofs}

\begin{proof}[Proof of Lemma~\ref{lemma:basic_props}]
The first property can be easily proven by noting that there can be no synergy
in a univariate system: i.e. if $n=1$, then $\mathcal{S}^{(k)} = \emptyset$ and
therefore $\texttt{Syn}^{(k)}(\bX; \bY) = 0$.

To prove the second property, let us assume that there exists a function
$g(\cdot)$ such that $g(X_t^j) = V_t$ for some $j$. Then, one can show that
\begin{align}
\texttt{Un}^{(1)}(V_t;\bm X_{t'}|\bX) &\leq I(V_t;\bm X_{t'}|\bX) \label{eq:no_g_1}\\
&\leq I(V_t;\bm X_{t'}|X_t^j) \\
&\leq H(V_t|X_t^j) \\
&=0~,
\end{align}
where \eqref{eq:no_g_1} is due to the non-negativity property of
$\texttt{Un}^{(k)}$ introduced in Section~\ref{sec:pid}. This shows that $V_t$
cannot exhibit emergent behaviour.

\end{proof}

\begin{proof}[Proof of Theorem~\ref{prop:emergence_syn}]
If $\texttt{Syn}^{(k)}(\bX; \bm X_{t'}) > 0$, then by Lemma~\ref{lemma:unsyn} it is
clear that the feature $V_t = \bm X_t$ exhibits causal emergence.

To prove the converse, note that all supervenient features follow the Markov
chain structure $V_t-\bX - \bm X_{t'}$. Therefore, for any supervenient feature
$V_t = f(\bm X_t)$ the following holds:
\begin{align}
0 \leq \texttt{Un}^{(k)}(V_t; \bm X_{t'} | \bm X_t ) 
&\leq 
\texttt{Un}^{(k)}(\bm X_t; \bm X_{t'} | \bm X_t ) \label{eq:dpi_uni33} \\
&=
\texttt{Syn}(\bm X_t; \bm X_{t'})~, \label{eq:so_so_cool}
\end{align}
where \eqref{eq:dpi_uni33} is due to the data processing inequality of the
unique information (c.f. Section~\ref{sec:pid}), and \eqref{eq:so_so_cool} is
due to Lemma~\ref{lemma:unsyn}. Using this result, is clear that if
$\texttt{Syn}^{(k)}(\bX; \bm X_{t'}) = 0$ then $\texttt{Un}^{(k)}(V_t; \bm
X_{t'} | \bm X_t) = 0$ for any superventient feature $V_t$.
\end{proof}

The above proof implies that the system has emergent behaviour if and only if
the system as a whole seen as a feature (i.e. $V_t=\bX$) is causally emergent.
Note that the fact that this trivial feature is helpful for detecting the
presence of emergence doesn't imply that it is an appropriate way of
representing it, as in most cases also carries non-interesting information, and
in practice one may be interested in features that exhibit emergence but are
shorter to describe than the microstate of the system, i.e. $H(V_t) < H(\bX)$.

\begin{proof}[Proof of Theorem~\ref{prop:downward}]
Let us first assume that there exists a supervenient feature $V_t$ such that
$\texttt{Un}^{(k)}(V_t;\bm X_{t'}^{\alpha} | \bm X) >0$ for some $\alpha :
|\alpha| = k$. Then, one can find that
\begin{align}
  0 & < \texttt{Un}^{(k)}(V_t; \bm X_{t'}^{\alpha} | \bX ) \nonumber\\
  &\leq \texttt{Un}^{(k)}(\bX; \bm X_{t'}^{\alpha} | \bX ) \label{eq:un2}\\
  &= \texttt{Syn}^{(k)}(\bX; \bm X_{t'}^{\alpha}) \label{eq:un3} \\
  &\leq \mathcal{D}^{(k)} (\bX; \bm X_{t'})~. \label{eq:un4}
\end{align}
Above, \eqref{eq:un2} is a consequence of the data processing inequality of the
unique information, \eqref{eq:un3} comes from Lemma~\ref{lemma:unsyn}, and
\eqref{eq:un4} is from the properties of $\mathcal{D}^{(k)}$ stated in
Section~\ref{sec:taxonomy}.

To prove the converse, let us assume that all supervenient features $V_t$
satisfy $\texttt{Un}^{(k)}(V_t;\bm X_{t'}^{\alpha} | \bm X_t) =0$ for all $k$
and $|\alpha| = k$. In particular, this is true for the feature $V_t=\bX$.
Then, another application of Lemma~\ref{lemma:unsyn} shows that
\begin{align}
  \mathcal{D}^{(k)} (\bX; \bm X_{t'})
  &\leq \sum_{|\alpha|=k}\texttt{Syn}^{(k)}(\bX; \bm X_{t'}^{\alpha}) \label{eq:un5}\\
  &= \sum_{|\alpha|=k} \texttt{Un}^{(k)}(\bX; \bm X_{t'}^{\alpha} | \bX ) \nonumber \\
  & = 0~. \nonumber
\end{align}
Above, \eqref{eq:un5} is a consequence of the properties of
$\mathcal{D}^{(k)}$.

\end{proof}

\begin{proof}[Proof of Theorem~\ref{prop:decoupled}]

We begin by proving that a system has a causally decoupled feature iff
$\mathcal{G}^{(k)}(\bX; \bY) > 0$. This proof follows a similar structure to
that of Theorem~\ref{prop:downward} for downward causation.

Let us first assume that there exists a supervenient feature $V_t$ with
$\texttt{Un}^{(k)}(V_t; V_{t'} | \bX, \bY) > 0$. Then,
\begin{align}
  0 &< \texttt{Un}^{(k)}(V_t; V_{t'} | \bX, \bY) \nonumber \\
    &\leq \texttt{Un}^{(k)}(\bX; \bY | \bX, \bY) \label{eq:decoup2} \\
    &= \mathcal{G}^{(k)}(\bX; \bY) \label{eq:decoup3} ~ .
\end{align}
Above, \eqref{eq:decoup2} can be obtained through a combination of the data
processing inequalities of the unique information and the synergy, as well as
Lemma~\ref{lemma:unsyn}; and \eqref{eq:decoup3} is an application of
Corollary~\ref{cor:unsyn_phiid}.

To prove the converse, assume that for all supervenient features
$\texttt{Un}^{(k)}(V_t; V_{t'} | \bX, \bY) = 0$. As before, this also includes
$V_t = \bX$. Therefore, applying Corollary~\ref{cor:unsyn_phiid} we arrive at
$\mathcal{G}^{(k)}(\bX; \bY) = \texttt{Un}^{(k)}(\bX; \bY | \bX, \bY) = 0$,
which concludes the proof.

We now move on to prove the results on perfectly causally decoupled systems, where
$\mathcal{G}^{(k)}(\bX;\bY) >0$ and $\mathcal{D}^{(k)}(\bX;\bY) = 0$. As
$\texttt{Syn}^{(k)}(\bX;\bY) = \mathcal{G}^{(k)}(\bX;\bY) +
\mathcal{D}^{(k)}(\bX;\bY) > 0$, thanks to Theorem~\ref{prop:emergence_syn}
this is equivalent to the existence of at least one emergent feature $V_t$.
Additionally, due to Theorem~\ref{prop:downward}, $\mathcal{D}^{(k)}(\bX;\bY) =
0$ is equivalent to $\texttt{Un}^{(k)}(V_t;\bm X_{t'}^{\alpha}|\bX) = 0$ for
all emergent features and $\alpha : |\alpha| = k$

\end{proof}

\section{Mathematical properties of emergence criteria}
\label{app:proofs2}

In this appendix we provide the necessary proofs linking the practical criteria
for emergence in Eqs.~\eqref{eq:emergence_psi} with the definitions in
Section~\ref{sec:formal_theory}. For completeness, we provide formulae of
$\Psi$, $\Delta$, and $\Gamma$ for arbitrary emergence order $k$:
\begin{subequations}
\begin{align}
\Psi_{t,t'}^{(k)}(V)&\coloneqq I(V_t; V_{t'}) - \sum_{|\alpha|=k} I(\bm X_{t}^{\alpha}; V_{t'} )~, \\
\Delta_{t,t'}^{(k)}(V)&\coloneqq \max_{|\alpha|=k} \left( I(V_t; \bm X^{\alpha}_{t'}) - \sum_{|\beta|=k} I(\bm X_{t}^{\beta}; \bm X^{\alpha}_{t'} ) \right)~, \\
\Gamma_{t,t'}^{(k)}(V)&\coloneqq \max_{|\alpha|=k} I(V_t; \bm X_{t'}^{\alpha})~.
\end{align}
\end{subequations}

\begin{proof}[Proof of Proposition~\ref{prop:emergence_psi}]

To start, note that a direct calculation shows that
\begin{align}
  \Psi_{t,t'}^{(k)}(V) 
  & \leq I(\bX; V_{t'}) - \sum_{|\alpha|=k} I(\bm X_{t}^{\alpha}; V_{t'} ) \label{eq:psi1}\\ 
  & \leq \texttt{Syn}^{(k)}(\bX; V_{t'}) \label{eq:psi2} \\
  & \leq \texttt{Syn}^{(k)}(\bX; \bm X_{t'}) \label{eq:psi3} ~.
\end{align}
Above, \eqref{eq:psi1} is due to the data processing inequality applied over
the Markov chain $V_t-\bX-\bY-V_{t'}$; \eqref{eq:psi2} is due to the
whole-minus-sum property of the synergy; and \eqref{eq:psi3} is due to the data
processing inequality of the synergy. Therefore, it is clear that if $
\Psi_{t,t'}^{(k)}(V) >0 $ for some feature $V_t$ then $\texttt{Syn}^{(k)}(\bX;
\bY) > 0$. This, combined with Theorem~\ref{prop:emergence_syn}, guarantees
that the system exhibits causal emergence.

To check the condition for downward causation, a direct calculation shows that,
for some $|\alpha|=k$,
\begin{align}
  \Delta_{t,t'}^{(k)}(V) 
  & \leq I(\bX; \bm X_{t'}^{\alpha}) - \sum_{|\bm \beta|=k} I(\bm X_t^{\beta}; \bm X^{\alpha}_{t'}) \label{eq:d1}\\ 
  & \leq \texttt{Syn}^{(k)}(\bX; \bm X^{\alpha}_{t'}) \label{eq:d2} \\
  & \leq \mathcal{D}^{(k)}(\bX; \bm X_{t'}) \label{eq:d3} ~.
\end{align}
Above, \eqref{eq:d1} is due to the data processing inequality applied over the
Markov chain $V_t-\bX-\bm X_{t'}^{\alpha}$; \eqref{eq:d2} to the
whole-minus-sum property of the synergy; and \eqref{eq:d3} to the
properties of $\mathcal{D}^{(k)}$. From here, is clear that if
$\Delta_{t,t'}^{(k)}(V) >0 $ for some feature $V_t$ then
$\mathcal{D}^{(k)}(\bX; \bY) > 0$. This, combined with
Theorem~\ref{prop:emergence_syn}, guarantees that the system exhibits downward
causation.

Finally, for the condition for causal decoupling it is sufficient to note that
\begin{align}
\Gamma_{t,t'}^{(k)}(V) &= \max_{|\alpha|=k} I(V_t; \bm X_{t'}^{\alpha}) \\
&\geq \texttt{Un}^{(k)}(V_t; \bm X_{t'}^{\alpha} | \bX ) \geq 0~,
\end{align}
where the inequality is due to the bounds of the unique information.
\end{proof}

\begin{proof}[Proof of Proposition~\ref{prop:autoemergence}]
Let us consider $V_t$ a $k$-synergistic observable. Due to stationarity, the
fact that $V_t\in\mathcal{C}_k(\bm X_t)$ implies that
$V_{t'}\in\mathcal{C}_k(\bm X_{t'})$. Using this fact, and noting that $V_t-\bm
X_t - \bm X_{t'} - V_{t'}$ is a Markov chain, it is direct to check that
\begin{equation}\label{eq:g_is_big}
\mathcal{G}_\star^{(k)}(\bX; \bY) \geq I(V_t ; V_{t'}) > 0~,
\end{equation}
Additionally, since $\mathcal{D}_\star^{(k)}(\bX; \bY) \geq 0$,
Eq.~\eqref{eq:g_is_big} implies that $\texttt{Syn}_\star^{(k)}(\bX; \bY) >0$,
proving the desired result.
\end{proof}

\section{Simulation details}
\label{app:sims}

Let us focus first on the Game of Life (GoL) simulations in
Figure~\ref{fig:gol}. For the initial state, two particles of three fixed types
(nothing, a glider, or a lightweight spaceship) were selected at random, and
placed at random positions in a 15x15 square cell array. The GoL evolution rule
was run for 1000 steps, which, in most cases (judged by visual inspection) was
enough for the system to settle on a stable configuration -- which was
typically either nothing, a small number of static structures, or a small
number of particles in non-colliding tracks. To compute the emergent feature
$V_{t'}$, particles were detected by simply pattern-matching the resulting
system state against the known shapes of each particle. We considered five
categories: still lifes, oscillators, gliders, lightweight spacesphips, or
nothing.

For static structures, we used a single symbol to represent all still lifes,
and a single symbol for all oscillators. This particle detector was found to be
very effective, with only 2\% of runs resulting in unrecognised particles. A
total of \num{5e4} independent runs were simulated and, due to the high number
of possible states of $V_t$, to reduce bias we used the quasi-Bayesian
estimator by Archer \emph{et al.}~\cite{archer2013bayesian}.

For the boids simulation, $N=10$ boids were simulated on a torus of side length
$L=200$. Boids are initialised with random positions, speeds, and head angles,
and at each timestep each boid $i=1, ..., N$ is updated according to the
equations:
\begin{align*}
  x^i_{t'} &= x^i_t + s^i \cos (\alpha^i_t) \\
  y^i_{t'} &= y^i_t + s^i \sin (\alpha^i_t) \\
  \alpha^i_{t'} &= \alpha^i_t + a_1 \theta_1 + a_2 (\pi + \theta_2) + a_3 \theta_3 ~ ,
\end{align*}
Where $\theta_1$ is the bearing to the flock's center of mass, $\theta_2$ the
bearing to the nearest boid, and $\theta_3$ the mean alignment of all boids
within a 20 unit radius. The scalars $a_1,a_2,a_3$ are the aggregation,
avoidance, and alignment parameters, respectively.

The results in Figure~\ref{fig:boids} were obtained averaging $\Psi$ across 25
independent runs of 5000 timesteps each, keeping $a_1 = 0.15$, $a_3 = 0.25$
fixed across all simulations. To compute $\Psi$, we pre-processed the
trajectories using the same procedure as Seth~\cite{seth2010measuring} (i.e.
each boid was described by its distance to the center of the environment and
all time series were first-order differenced), and information-theoretic
quantities were computed using the non-parametric estimator implemented in the
JIDT toolbox~\cite{kraskov2004estimating,lizier2014jidt} using a dynamic
correlation exclusion window of 10 samples~\cite{kantz2004nonlinear}.

\section{ECoG preprocessing and decoding}
\label{app:ecog}

ECoG signals were preprocessed following the steps presented in the original
publication: pecifically, the data was notch-filtered to remove line noise and
band-pass filtered from \SIrange{0.1}{600}{\hertz}, and then downsampled to the
same sampling frequency as the MoCap data
(\SI{120}{\hertz})~\cite{chao2010long}. Then, data were divided into a 60/40
train/test split. To build the predictor, the training set was first
standardised to zero mean and unit variance, and then it was
dimensionality-reduced with a 20-component Partial Least Squares (PLS)
regression using the three dimensions of the MoCap as dependent variables. We
then trained a non-linear Support Vector Machine (SVM)~\cite{bishop2006pattern}
using a squared exponential kernel predicting the MoCap data from the
20-dimensional PLS scores. SVM and kernel hyperparameters were optimised
through 5-fold cross-validation, and a model with the best hyperparameter
configuration was trained on the full training set. This cross-validation and
optimisation procedure was repeated for each dimension of the MoCap data,
resulting in three separate SVMs~\footnote{Note that the accuracy of this model
is lower than the model presented in Ref.~\cite{chao2010long}, for a precise
reason: to predict position at time $t$, Chao \emph{et al.} use a wavelet
transform to obtain features that have represent ECoG from
$t-$\SI{1.1}{\second} to $t$. Unlike Chao \emph{et al.}, we work under the
constraint of using supervenient features: for the axioms of the theory (in its
current form) to hold, $V_t$ needs to be a function \emph{of} \bX \emph{only}.
This limits the possibilities for feature extraction, and naturally may result
in lower-accuracy decoders.}.

The computation of $\Psi$ was performed on the held-out test data set. ECoG
signals were standardised and projected onto the PLS latent space (using the
means, standard deviations, and projection matrix obtained from the training
data), and mutual information terms involved in the computation of $\Psi$ were
calculated using the open-source JIDT toolbox~\cite{lizier2014jidt}.


\end{document}